\documentclass[11pt]{article}

\usepackage[margin=2.5cm]{geometry}
\usepackage[hidelinks]{hyperref}
\usepackage{authblk}
\usepackage{physics}
\usepackage{amsmath}
\usepackage{amssymb}
\usepackage{amsthm}

\theoremstyle{plain}
\newtheorem{theorem}{Theorem}[section]
\newtheorem{lemma}[theorem]{Lemma}
\newtheorem{corollary}[theorem]{Corollary}
\theoremstyle{definition}

\newtheorem{definition}[theorem]{Definition}

\DeclareMathOperator{\TD}{TD}

\DeclareMathOperator{\Adv}{Adv}
\DeclareMathOperator{\Fid}{F}
\DeclareMathOperator*{\argmax}{argmax}

\newcommand{\Pclass}{\textsf{P}}
\newcommand{\NPclass}{\textsf{NP}}
\newcommand{\BPPclass}{\textsf{BPP}}
\newcommand{\BQPclass}{\textsf{BQP}}
\newcommand{\BQPqpolyclass}{\textsf{BQP/qpoly}}
\newcommand{\Ppolyclass}{\textsf{P/poly}}

\newcommand{\As}{\mathsf{A}}
\newcommand{\Dc}{\mathcal{D}}

\newcommand{\N}{\mathbb{N}}
\newcommand{\R}{\mathbb{R}}
\newcommand{\E}{\mathbb{E}}

\title{Pseudo-Entanglement is Necessary for EFI Pairs}

\author[]{Manuel Goul\~ao}
\author[]{David Elkouss}
\affil[]{Okinawa Institute of Science and Technology Graduate University, Japan}
\affil[]{\url{{manuel.goulao,david.elkouss}@oist.jp}}

\date{}

\begin{document}

\maketitle

\begin{abstract}
Regarding minimal assumptions, most of classical cryptography is known to depend on the existence of One-Way Functions (OWFs).
However, recent evidence has shown that this is not the case when considering quantum resources.
Besides the well known unconditional security of Quantum Key Distribution, it is now known that computational cryptography may be built on weaker primitives than OWFs, e.g., pseudo-random states~\cite{C:JLS18}, one-way state generators~\cite{pre:MY23}, or EFI pairs of states~\cite{ITCS:BCQ23}.

We consider a new quantum resource, pseudo-entanglement, and show that the existence of EFI pairs, one of the current main candidates for the weakest computational assumption for cryptography (necessary for commitments, oblivious transfer, secure multi-party computation, computational zero-knowledge proofs), implies the existence of pseudo-entanglement, as defined by~\cite{pre:ABF+22,pre:ABV23} under some reasonable adaptations.
We prove this by constructing a new family of pseudo-entangled quantum states given only EFI pairs.

Our result has important implications for the field of computational cryptography.
It shows that if pseudo-entanglement does not exist, then most of cryptography cannot exist either. Moreover, it establishes pseudo-entanglement as a new minimal assumption for most of computational cryptography, which may pave the way for the unification of other assumptions into a single primitive.
Finally, pseudo-entanglement connects physical phenomena and efficient computation, thus, our result strengthens the connection between cryptography and the physical world.
\end{abstract}

\newpage

\section{Introduction}
In classical cryptography, the existence of most primitives is established from connections with complexity theoretic assumptions.
In fact, if One-Way Functions (OWFs) exist, then \(\Pclass \neq \NPclass\).
Moreover, the existence of OWFs is both necessary and sufficient for most symmetric primitives (and also signature schemes) --- usually called the \textit{Minicrypt} world --- and is also implied by more complex primitives such as public-key cryptography --- the \textit{Cryptomania} world.
This makes OWFs a minimal assumption, whose existence is necessary for most of classical cryptography to exist~\cite{IPL:G90}.
This hierarchical structure based on the average hardness of problems in \(\Pclass\) and \(\NPclass\) is usually called the worlds of computational complexity theory~\cite{CT:I95}.

In quantum cryptography, the existence of such a hierarchy is not known.
Indeed, unconditionally-secure Quantum Key Distribution (QKD)~\cite{TCS:BB84} is possible using quantum information resources, while classical key-establishment procedures require public-key cryptography.
In addition, in~\cite{EC:GLSV21,C:BCKM21}, it is shown that primitives such as oblivious transfer and secure multi-party computation may be constructed from just OWFs and quantum resources, while classically public-key cryptography is required.
Decisively, recent results~\cite{C:AQY22,C:MY22,AC:Y22,ITCS:BCQ23,pre:MY23,pre:LMW23} show that many cryptographic primitives may be achieved under weaker assumptions than OWFs, by making weaker computational assumptions that are intrinsically based on quantum information features.
The first proposed candidate to build cryptography from weaker (than OWFs) assumptions was Pseudo-Random States (PRSs)~\cite{C:JLS18}, efficiently preparable quantum states that are computationally indistinguishable from Haar-random states, and for which OWFs are sufficient but strictly not necessary (in the relativizing setting)~\cite{ITCS:K21,C:AQY22,STOC:KQST23,pre:LMW23}.
Indeed, the existence of PRSs seems to be a problem independent of the question of \(\Pclass \stackrel{?}{=}\NPclass\)~\cite{STOC:KQST23} and even independent of all traditional computational complexity~\cite{pre:LMW23}.
On top of that, new primitives were introduced, namely EFI pairs~\cite{ITCS:BCQ23} and OWSGs~\cite{pre:MY23}, which seem to be even weaker assumptions that PRSs.
An important consequence of these results is that, by leveraging these new intrinsically quantum resources, computational cryptography might still exist even if \(\Pclass = \NPclass\).

Currently, no single primitive is known to be necessary for all of quantum computational cryptography (in the way that OWFs are for classical cryptography).
Still, one of the main candidates is EFI pairs of quantum states~\cite{ITCS:BCQ23}.
This primitive is composed of two quantum states that are Efficiently generated, statistically Far, and computationally Indistinguishable, and are an analogous counterpart to the classical EFI pairs of distributions (EFID pairs)~\cite{IPL:G90}.
EFI pairs are implied by the existence of commitment schemes, oblivious transfer, secure multi-party computation, and zero-knowledge proofs for non-trivial languages~\cite{ITCS:BCQ23}.
Notwithstanding, alternative candidates also exist, such as One-Way State Generators (OWSGs)~\cite{C:MY22,pre:MY23}.
The connection between these primitives is not currently known, but recent results give some evidence about their relation.
In particular, \cite{pre:KT23} shows that restricting the states generated by the OWSGs to pure states means that OWSGs imply EFI pairs.
Also, \cite{pre:MMWY24} further extends the connection by showing that, for mixed states, OWSGs with exponential security imply EFI pairs.
These evidences seem to hint that EFI pairs are a weaker assumption than OWSGs, but the general relation between these primitives is still unknown.

The concept of pseudo-randomness, i.e., that computationally bounded machines are unable to distinguish pseudo-random and truly random information, has long been established as a fundamental topic in computer science.
Recently, there has been advances in the study of physical resources in the presence of computationally bounded observers.
In this setting, analogously to the classical concept of pseudo-randomness, quantum states with a low amount of some resource are computationally indistinguishable from states with a high amount of that resource.
These include the study of pseudo-entanglement~\cite{pre:ABF+22,pre:ABV23}, pseudo-magic~\cite{PRL:GLG+24}, and pseudo-coherence and pseudo-purity and pseudo-imaginarity~\cite{pre:HBK24}.
While the study of some of these quantities has so far been restricted to the case of pure states, some separation gaps have also been analyzed for pseudo-random density matrices~\cite{pre:BMB+24}.
Moreover, interesting connections with cryptography start to be established, enabling separations between computational and information-theoretic entanglement measures~\cite{pre:ABV23}, and pseudo-magic ensembles implying EFI pairs~\cite{PRL:GLG+24}.

\medskip

In this work, we will focus on EFI pairs, and show that their existence implies the existence of families of pseudo-entangled states~\cite{pre:ABF+22,pre:ABV23}.
Intuitively, pseudo-entanglement is a property exhibited by a family of quantum states that has low entanglement (as quantified by some measure), while they look like higher entanglement states to any observer that is computational limited to performing Quantum Polynomial Time (QPT) operations.
In a way, this notion generalizes the concept of pseudo-randomness in classical computation and the indistinguishability properties of certain distribution ensembles, to entanglement in quantum computation.
This opens the possibility for research in new fields of study such as computational entanglement, the approach to the use of the physical resource of entanglement viewed by the lens of computational complexity, generalizing its more common handling as an information-theoretic-only resource.
We show that the existence of EFI pairs implies the existence of pseudo-entanglement, when considering a description of pseudo-entanglement inspired by the definitions of both~\cite{pre:ABF+22} and~\cite{pre:ABV23} through some reasonable adaptations.
We achieve a polynomial entanglement gap between two indistinguishable families of states, where the pseudo-entangled states have no entanglement whatsoever (they are separable), while the high-entanglement states have entanglement that scales polynomially with the size of the states.

Currently, some pseudo-entanglement constructions have been proposed based on the frameworks of~\cite{pre:ABF+22} and~\cite{pre:ABV23}.
Below, we compare the required assumptions of these constructions, but of particular relevance is that all current pseudo-entanglement proposals are also PRSs, which seems a stronger assumption than EFI pairs.
\begin{itemize}
    \item In~\cite{pre:ABF+22}, the first pseudo-entanglement formalism is proposed. The constructed pseudo-entangled states are in the form of ``subset phase states'' that are also PRSs, and these states are prepared assuming the existence of quantum-secure pseudo-random permutations and pseudo-random functions, both of which are equivalent to OWFs.
    \item In~\cite{pre:ABV23}, another pseudo-entangled formalism is introduced. Also in this work, a construction of pseudo-entangled states that are also PRSs is given, assuming the existence of pseudo-random permutations, which is equivalent to assuming the existence of OWFs.
    \item In~\cite{pre:GB23}, it is shown that pseudo-random ``subset states'' are PRSs and exhibit pseudo-entanglement, relaxing the requirement of pseudo-random relative phases that were present in the ``subset phase states'' of~\cite{pre:ABF+22} (this pseudo-entanglement formalism is used).
    This implies that the construction gives pseudo-entangled states that are also PRSs, but only require the existence of pseudo-random permutations (equivalent to OWFs).
    \item For completeness, we consider~\cite{pre:BFG+23}, an extension of the pseudo-entanglement framework of~\cite{pre:ABF+22}, called ``public-key pseudo-entanglement'', where the circuits used to prepare the states are public.
    This extension may also be made to the pseudo-entanglement framework of~\cite{pre:ABV23} by allowing the distinguisher to have access to the keys~\cite{pre:ABV23}.
    The construction holds on the existence of ``lossy functions'', a stronger assumption that implies OWFs.
    While the constructed pseudo-entangled states are not required to be PRSs, they can be made so by the use of pseudo-random functions (equivalent to OWFs).
    \item In terms of arbitrary pseudo-entanglement constructions, \cite{pre:GIKL24} shows that the complexity of constructing pseudo-entangled states in terms of the number of non-Clifford gates implies that a gap of \(t\) bits of entropy requires \(\Omega(t)\) non-Clifford gates to prepare.
    Their result is constructed on the formalism of~\cite{pre:ABF+22}, but it does not limit this restriction to pseudo-entangled states that are also PRSs.
\end{itemize}
All proposed pseudo-entanglement constructions rely on the existence of OWFs or stronger assumptions.
In opposition, our proposal is a substantially different candidate, as it relies only on the weaker assumption of EFI pairs, albeit for an adapted definition of pseudo-entanglement.
The most relevant differences of our definitions are that we consider mixed states and appropriate entanglement measures, we only consider a bipartite setting with entanglement across a single cut, and we do not require efficient distillation or preparation of the states of the reference (high-entanglement) family.

Our result has various implications for the field of theoretical cryptography.
First, an immediate corollary is that if families of pseudo-entangled states do not exist, then most cryptographic constructions are impossible to realize (as they imply pseudo-entanglement).
Second, it gives a new candidate for a minimal assumption necessary for cryptography, since it is either weaker or equivalent to EFI pairs.
Third, pseudo-entanglement is, in essence, connecting the properties of the physical world and efficient computation, which is argued in~\cite{pre:ABF+22,pre:ABV23} with a connection to the AdS/CFT correspondence; here, we create a bridge between the existence cryptography and physical phenomena.

\subsection{Contributions}\label{sec:contributions}

\paragraph{Existence of pseudo-entanglement is necessary for the existence of EFI pairs.}
Our main theorem shows that if EFI pairs exist, then, so must exist families of pseudo-entangled states, demonstrating that the existence of pseudo-entanglement is a weaker or equal assumption than EFI pairs.
Before our work, no connection was known relating these two primitives, and this relation is left as an open question in~\cite{pre:ABV23}.
Moreover, EFI pairs have been shown to be a fundamental assumption in cryptography.
In particular, the existence of EFI pairs is necessary for the existence of commitments, oblivious transfer, secure multi-party computation, and computational zero-knowledge proofs for non-trivial languages.
Consequently, this means that if EFI pairs do not exist, then these primitives are impossible to realize.
Our result further extends this insight, and says that if families of pseudo-entangled states do not exist, then most of cryptography, namely the aforementioned primitives and those implied by them, does not exist. 

\paragraph{New construction of pseudo-entanglement.}
In order to show that EFI pairs imply pseudo-entanglement, we give an explicit construction of new families of states that are pseudo-entangled.
It is noteworthy that our construction yields families of pseudo-entangled states that are not PRSs~\cite{C:JLS18} (states indistinguishable from Haar-random states).
Indeed, previously, all proposed families of pseudo-entangled states were also instances of PRSs~\cite{pre:ABF+22,pre:ABV23}.
This is especially meaningful as it seems that the existence of PRSs is a stronger assumption than the existence of EFI pairs~\cite{ITCS:BCQ23}, and therefore, we are exhibiting a pseudo-entanglement construction that may exist even if PRSs do not.
Implicitly, our construction gives the weakest known construction of pseudo-entanglement in terms of assumption requirements.

\paragraph{Polynomial amplification of pseudo-entanglement.}
We first show that the existence of pseudo-entanglement is necessary for the existence of EFI pairs, and we do this by providing a construction that exhibits a difference of exactly 1 e-bit (i.e., 1 EPR pair) between the entanglement of the two families of the pseudo-entanglement structure.
We extend this result with a lemma that says that any gap between families of pseudo-entangled states may be polynomially amplified, by taking the tensor of polynomially many states from each family.
Since we are not handling PRSs, this result (as far as we are aware) was not previously known in the literature, and we found it relevant to show that it is possible to amplify the entanglement gap without breaking the computational indistinguishability requirement.
Moreover, we may then directly take as a corollary that our construction with gap-1 pseudo-entanglement may be amplified to a polynomial gap (in the size of the problem).
Accordingly, this corollary yields the first construction of families of pseudo-entangled states with polynomial entanglement separation and that do not require PRSs.

\paragraph{New candidate for minimal assumption for cryptography.}
Arguably, the most impactful consequence of our result is the introduction of the new minimal assumption required for the existence of computational cryptography known to this date --- pseudo-entanglement.
Now, it becomes known that most cryptographic primitives, such as private-key (Minicrypt world) and public-key (Cryptomania world) cryptography, which imply OWFs, cannot exist without pseudo-entanglement also existing.
Also, even weaker primitives, such as commitments, oblivious transfer, secure multi-party computation, computational zero-knowledge proofs (which were known to imply EFI pairs) are impossible to realize if pseudo-entanglement does not exist.
While the impossibility of EFI pairs would also have these devastating consequences for cryptography, the consideration of pseudo-entanglement not only seems to be a weaker assumption, but also provides new tools that open up the unification into a single assumption of other primitives.
Moreover, its intrinsic physical nature directly links computational hardness to the properties of the physical world, which bridges these two fields and may give new insights to the interpretation of cryptography as physical phenomena.

\subsection{Open Questions}\label{subsec:openquestions}
The first and perhaps most natural question that is left open after our work is whether the two primitives, EFI pairs and pseudo-entanglement, are actually equivalent.
At first sight, it seems that the other direction of the implication is not true, since the most natural construction would sample each element of the EFI pair as a state from each of the pseudo-entanglement families (one element of the pair from the ``low-entanglement'' and the other from the ``high-entanglement'' family), but the ``high-entanglement'' family does not have to be efficient to sample, invalidating such construction.
Also, by simply taking two states from the ``low-entanglement'' (the pseudo-entangled) family, showing the required properties of the EFI pairs also seems difficult, as the only special property that these states have is that they are indistinguishable from differently structured and more entangled states.
Nevertheless, such reduction would be interesting to further establish the relation between EFI pairs and other possible computational or physical assumptions.

Another fundamental question that is left open is whether pseudo-entanglement may be used to unify other ``minimal assumption'' candidates into the same level, and establish what might be interpreted as the quantum Minicrypt analog (sometimes called the MiniQCrypt).
In the same note, it would be interesting to find any complementary computational (or information-theoretic) problems that do not require the existence of pseudo-entanglement, i.e., what exists below this level.
More pragmatically, it would be of interest to find what cryptographic primitives could still be realizable, even if pseudo-entanglement does not exist, besides information-theoretic cryptography.
For this, it is fundamental to first consider what is the most relevant definition of pseudo-entanglement to fit the scope of complexity theory and consequently of cryptography.

Lastly, being able to link computational hardness as some emerging phenomenon from physics may shed light on complexity theoretic problems for which a logical proof might not even exist, while some separation of complexity classes could still be established as a property of nature.
Therefore, connecting pseudo-entanglement with concrete complexity theoretic results or with physical phenomena would reveal significant insights about cryptography as a computational and physical resource.
In particular, it would be interesting to show that there exists some physical phenomenon that could imply or be implied by the existence of pseudo-entanglement, and this would widen the connection between the properties of physics and of efficient computation and cryptography.
More recently, this relation has been studied through the lens of the AdS/CFT correspondence, and expanding on this connection might lead to significant advances in both research areas.

\subsection{Technical Overview}
In this section, we provide a self-contained, technical, yet intuitive view of our results, which are thoroughly formalized and covered in detail in latter sections.
We start by giving an overview of the main results.
We then give intuitive descriptions of important concepts, such as EFI pairs and pseudo-entanglement.
Also, we describe the rationale behind the design of our new families of pseudo-entangled states constructed from EFI pairs.
Finally, we explain the key techniques used to prove our main results, namely Theorem~\ref{thm:efi2pe} and Corollary~\ref{cor:efi2pe-poly}.\\

Our main technical results are presented in two different parts:
\begin{enumerate}
    \item The main results, presented in Sections~\ref{sec:efi2pe} and~\ref{sec:quantumk}, show that pseudo-entanglement is necessary for EFI pairs.
    This is shown by constructing two families of  bipartite mixed states: one pseudo-entangled family \(\{\psi\}\), which is cheap to construct when considering entanglement, in particular it is separable and thus may be constructed using only Local Operations and Classical Communication (LOCC); and, one family of entangled states \(\{\phi\}\), with $1$ e-bit of entanglement (i.e., one EPR pair).
    Notwithstanding, these two families of states are indistinguishable to any party that may even hold polynomially many copies of the full states locally, but is restricted to performing QPT operations.
    To construct these two families of states, we only assume that we have access to an algorithm generating EFI pairs, implying that \emph{pseudo-entanglement is necessary for  EFI pairs to exist}.
    \item A corollary of our main result, given in Section~\ref{sec:amplification}, extends the previously achieved reduction (that shows the fundamental implication from EFI pairs to pseudo-entanglement) by widening the gap by a polynomial factor between the entanglement of the pseudo-entangled family \(\{\psi\}\) and its high-entangled counterpart \(\{\phi\}\).
    This is accomplished by first proving a general lemma that applies to any family of pseudo-entangled states and shows that the difference between the entanglement of the computationally indistinguishable families \(\{\psi\}\) and \(\{\phi\}\) may always be amplified by a polynomial factor (in the security parameter).
    Then, by direct application of this lemma to our main result, we conclude that \emph{polynomial-gap pseudo-entanglement is necessary for EFI pairs to exist}.\\
\end{enumerate}

We now give an informal statement of our main result (Theorem~\ref{thm:efi2pe}), and intuitively explain the required background below.
In particular, we will introduce the concepts of EFI pairs of states and pseudo-entanglement.
Then, we explain the basis of our construction and reason how it does fulfill the requirements to prove the theorem.

\begin{theorem}[EFI pairs $\Rightarrow$ Pseudo-Entanglement (Informal)]\label{thm:efi2pe-inf}
    If there exist EFI pairs of states \((\rho_{0,\lambda},\rho_{1,\lambda})\), then there exists pseudo-entanglement.
    I.e., there exist computationally indistinguishable (under polynomially many copies) families of quantum states, \(\{\psi_{AB}^\lambda\}_\lambda\) and \(\{\phi_{AB}^\lambda\}_\lambda\), with \(\hat{E}_C^\varepsilon\left(\{\psi_{AB}^\lambda\}_\lambda\right) = 0\) and \({E}_D^\varepsilon\left(\{\phi_{AB}^\lambda\}_\lambda\right) = 1\) (with negligible \(\varepsilon\)).
\end{theorem}

\paragraph{EFI pairs~\cite{ITCS:BCQ23}} are pairs of mixed states \((\rho_{0,\lambda},\rho_{1,\lambda})\), which are:
\begin{itemize}
    \item Efficiently preparable: there exists a QPT algorithm \(\As\) that on input \((1^\lambda,b)\) outputs \(\rho_{b,\lambda}\).
    \item Statistically far: \(\TD(\rho_{0,\lambda},\rho_{1,\lambda}) \approx 1\).
    \item Computationally indistinguishable: for all QPT distinguishers \(\mathcal{D}\), its advantage 
        \begin{equation*}
            \Adv_\Dc (\rho_{0,\lambda}, \rho_{1,\lambda} ) = \left| \Pr[1\leftarrow \mathcal{D}(\rho_{0,\lambda})] - \Pr[1\leftarrow \mathcal{D}(\rho_{1,\lambda})] \right|
        \end{equation*}
    is negligible.
\end{itemize}
As previously discussed, EFI pairs are a cryptographic primitive from which some cryptographic constructions (e.g., commitments, oblivious transfer, secure multi-party computation, computational zero knowledge proofs) are possible to realize without the need of OWFs.
This makes EFI pairs a weak requirement to build cryptography on, and arguably the current most promising candidate for minimal assumption for the existence of most of computational cryptography.

\paragraph{Pseudo-entanglement~\cite{pre:ABF+22,pre:ABV23}} is a property exhibited by a family of quantum states \(\{\psi_{AB}^\lambda\}_\lambda\) that have low entanglement, which is computationally indistinguishable from another reference family of quantum states \(\{\phi_{AB}^\lambda\}_\lambda\) with higher entanglement.\\
The first family, \(\{\psi_{AB}^\lambda\}_\lambda\), is ``cheap'' to construct when considering its entanglement.
To evaluate this, an appropriate entanglement measure must be used to quantify the entanglement of this family \(\{\psi_{AB}^\lambda\}_\lambda\), and the computational entanglement cost may be used as an extreme measure to upper bound the entanglement amount of the states.
The computational entanglement cost of the family \(\{\psi_{AB}^\lambda\}_\lambda\), 
\begin{equation*}
    \hat{E}_C^\varepsilon(\{\psi_{AB}^\lambda\}_\lambda) \leq c(\lambda),
\end{equation*}
is defined as an upper bound on the number of EPR pairs that the two parties holding systems \(A\) and \(B\) need to share, such that in an LOCC scenario restricted to QPT operations they can prepare the states \(\psi_{AB}^\lambda\).\\
The second family, \(\{\phi_{AB}^\lambda\}_\lambda\), has higher entanglement.
Again, an appropriate entanglement measure must be used to quantify the entanglement of this family, and the distillable entanglement may be used as an extreme measure to lower bound the entanglement amount of the states.
The distillable entanglement of the family \(\{\phi_{AB}^\lambda\}_\lambda\) is defined as the number of EPR pairs that the two parties holding systems \(A\) and \(B\) can distill in an LOCC scenario, meaning that they hold this amount of EPR pairs by the end of the protocol. We consider a lower bound on this quantity as
\begin{equation*}
    E_D^\varepsilon(\{\phi_{AB}^\lambda\}_\lambda) \geq d(\lambda).
\end{equation*}
Naturally, pseudo-entanglement is only interesting if \(c < d\).
Moreover, these quantities allow for an error \(\varepsilon\) (defined in terms of the fidelity of the output of the protocol with its desired output), and \(\varepsilon\) should be a negligible function in \(\lambda\).\\
The two families must be indistinguishable to any distinguisher \(\Dc\) that is restricted to performing efficient quantum operations (QPT), even when it is given polynomially (\(p(\lambda)\)) many copies of the full states (both partitions), i.e.,
\begin{equation*}
    \Adv_\Dc\left({\psi_{AB}^\lambda}^{\otimes p(\lambda)},{\phi_{AB}^\lambda}^{\otimes p(\lambda)}\right) = \left|\Pr[1\leftarrow\Dc\left({\psi_{AB}^\lambda}^{\otimes p(\lambda)}\right)] - \Pr[1\leftarrow\Dc\left({\phi_{AB}^\lambda}^{\otimes p(\lambda)}\right)]\right|
\end{equation*}
is negligible.\\
We remark that we make some adaptations to the definitions of pseudo-entanglement of~\cite{pre:ABF+22,pre:ABV23} to make them compatible with the concept of EFI pairs.
In particular, we consider mixed states and entanglement across a single partition (instead of pure states and entanglement across all cuts of a  sufficiently large number qubits, as in the definition of~\cite{pre:ABF+22}); and, we do not require the entanglement of the family \(\{\phi_{AB}^\lambda\}_\lambda\) to be efficiently distillable as these states are not required to be efficiently preparable (as considered in the definition~\cite{pre:ABV23}).
Nevertheless, in Section~\ref{sec:quantumk}, we provide a generalization to the LOCC setting of~\cite{pre:ABV23} where one party has available a quantum state that acts as a secret key (\(k_\lambda\)) that allows for efficient distillation of the entanglement of the family \(\{\phi_{AB}^\lambda\}_\lambda\), i.e., the computational distillable entanglement
\begin{equation*}
    \hat{E}_D^\varepsilon(\{k_\lambda,\phi_{AB}^\lambda\}_\lambda) \geq d(\lambda).
\end{equation*}

We are now ready to describe our construction of pseudo-entanglement, which depends exclusively on the existence of EFI pairs.
\paragraph{Construction of pseudo-entangled families of states:}
\begin{align}
    \psi_{AB}^\lambda &= \frac{1}{4}\left(\ketbra{\Phi^+}_{AB} + \ketbra{\Phi^-}_{AB}\right)\otimes ({\rho_{0,\lambda}}_A + {\rho_{1,\lambda}}_A ),\label{eq:psi1-inf}\\
    \phi_{AB}^\lambda &= \frac{1}{2}\left(\ketbra{\Phi^+}_{AB}\otimes {\rho_{0,\lambda}}_A + \ketbra{\Phi^-}_{AB}\otimes {\rho_{1,\lambda}}_A  \right).\label{eq:phi1-inf}
\end{align}
Our construction is inspired by the concept of false entropy~\cite{STOC:ILL89} and the equivalence reduction between false entropy and EFID pairs~\cite{IPL:G90}.
Indeed, a distribution ensemble \(X = \{X_i\}_i\) has false entropy~\cite{STOC:ILL89} if there exists another ensemble \(Y = \{Y_i\}_i\) that has higher entropy than \(X\), but both ensembles are computationally indistinguishable.
Clearly, this property may be seen as a classical counterpart to the new quantum notion of pseudo-entanglement~\cite{pre:ABF+22,pre:ABV23}.
Moreover, in~\cite{IPL:G90}, it is shown that false entropy implies (and is implied by) EFID pairs, pairs of distribution ensembles that are efficiently sampleable, statistically far, and yet computationally indistinguishable, the classical analog of EFI pairs~\cite{ITCS:BCQ23}.
To show the implication from EFID pairs to false entropy, one may create two ensembles \(W=\{W_i\}_i\) and \(Z=\{Z_i\}_i\) from \(X\) and \(Y\), such that
\begin{itemize}
    \item \(W_i = (B,X_i)\) with probability \(1/2\) and \(W_i = (B,Y_i)\) with probability \(1/2\), where \(B\) is an independent Bernoulli random variable with parameter \(p=1/2\).
    \item \(Z_i = (0,X_i)\) with probability \(1/2\) and \(Z_i = (1,Y_i)\) with probability \(1/2\).
\end{itemize}
Clearly, \(W\) has one more bit of entropy than \(Z\), as the added first bit to both \(X_i\) and \(Y_i\) acts as a label in case of \(Z\) but is independent in case of \(W\).
This is analogous to the construction of the states \(\psi_{AB}^\lambda\) and \(\phi_{AB}^\lambda\).
Lastly, we remark that (as also stated in Section~\ref{subsec:openquestions}) equivalence does not seem to be true in the quantum setting as it was in the classical setting.
Besides our intuition described in Section~\ref{subsec:openquestions}, the direction of the proof that shows that EFID pairs imply false entropy relies on the fact that EFID pairs are equivalent to pseudo-random generators and thus OWFs, which does not seem to be true for the quantum EFI pairs~\cite{ITCS:BCQ23}.\\

To prove our main result, informally stated in Theorem~\ref{thm:efi2pe-inf}, we show that the states \(\psi_{AB}^\lambda\) of the form of Equation~\eqref{eq:psi1-inf} and \(\phi_{AB}^\lambda\) of the form of Equation~\eqref{eq:phi1-inf}, directly constructed from EFI pairs \(\rho_{0,\lambda},\rho_{1,\lambda}\), have different entanglement.
Concretely, \(\psi_{AB}^\lambda\) has no entanglement and \(\phi_{AB}^\lambda\) has 1 e-bit of entanglement, while both families are computationally indistinguishable from one another even provided polynomially many copies.

\paragraph{The states \(\psi_{AB}^\lambda\) are separable.} 
The first register of \(\psi_{AB}^\lambda\) is a uniform mixture of two EPR pairs, which is a Bell diagonal state with probabilities equal to \(1/2\), and so is separable and may be constructed with only QPT operations by LOCC between systems \(A\) and \(B\).
The second register is a uniform mixture of the two elements of the EFI pair, but this is a local register to system \(A\).
Therefore, the states \(\psi_{AB}^\lambda\) may be constructed efficiently just from LOCC with QPT operations and without any shared EPR pairs. 
Thus, the computational entanglement cost of the family \(\{\psi_{AB}^\lambda\}_\lambda\) is zero, \begin{equation*}
    \hat{E}_C^0(\{\psi_{AB}^\lambda\}_\lambda) = 0.
\end{equation*}

\paragraph{The states \(\phi_{AB}^\lambda\) are entangled.} In \(\phi_{AB}^\lambda\), the EFI pair elements are coupled with a different EPR pair each, and act as a label to indicate whether the state sampled from the mixture is \(\ketbra{\Phi^+}\) (EFI label \(\rho_{0,\lambda}\)) or \(\ketbra{\Phi^-}\) (EFI label \(\rho_{1,\lambda}\)).
Indeed, this family of states may not be prepared from just LOCC between the bipartition \(A:B\) without previously shared entanglement. 
Moreover, these states hold 1 e-bit of entanglement, and the parties may distill 1 shared EPR pair across this bipartition.
For this, a party holding system \(A\) that may differentiate between \(\rho_{0,\lambda}\) and \(\rho_{1,\lambda}\) may communicate back (classically) the result to the party holding system \(B\). 
This distinguishing party must exist, since, by assumption of EFI pair, the trace distance between the elements of the pair is exponentially close to 1, so it has an overwhelming probability of correctly distinguishing. 
Thus, the entanglement distillation of the family \(\{\phi_{AB}^\lambda\}_\lambda\) is 1 with error \(\varepsilon\in O(2^{-\lambda})\),
\begin{equation*}
    E_D^{2^{-\lambda}}(\{\phi_{AB}^\lambda\}_\lambda) = 1.
\end{equation*}

\paragraph{The families \(\{\psi_{AB}^\lambda\}_\lambda\) and \(\{\phi_{AB}^\lambda\}_\lambda\) are computationally indistinguishable.} 
If no distinguisher \(\Dc\) has non-negligible advantage in distinguishing \(\rho_{0,\lambda}\) and \(\rho_{1,\lambda}\),
then no distinguisher \(\Dc'\) has non-negligible advantage in distinguishing \(\psi_{AB}^\lambda\) and \(\phi_{AB}^\lambda\).
I.e., for all non-uniform QPT distinguishers \(\Dc,\Dc'\) there exist two negligible functions \(\varepsilon,\varepsilon'\) such that
\begin{equation*}
    \Adv_\Dc(\rho_{0,\lambda},\rho_{1,\lambda})\leq \varepsilon(\lambda) \quad \Rightarrow \quad \Adv_{\Dc'}(\psi_{AB}^\lambda,\phi_{AB}^\lambda)\leq \varepsilon'(\lambda).
\end{equation*}
To prove this property, we show that the advantage of any distinguisher \(\Dc'\) trying to distinguish \(\psi_{AB}^\lambda\) and \(\phi_{AB}^\lambda\) is upper bounded by the advantage of any distinguisher \(\Dc\) distinguishing \(\rho_{0,\lambda}\) and \(\rho_{1,\lambda}\), which is negligible by assumption (of the existence of EFI pairs),
\begin{equation*}
    \Adv_{\Dc'}(\psi_{AB}^\lambda,\phi_{AB}^\lambda) \leq \frac{1}{2} \Adv_\Dc(\rho_{0,\lambda},\rho_{1,\lambda})\leq  \frac{1}{2} \varepsilon(\lambda).
\end{equation*}
Intuitively, there are two possibilities that may occur when a distinguisher \(\Dc'\) is trying to distinguish the mixed states (i.e., distributions) \(\psi_{AB}^\lambda\) and \(\phi_{AB}^\lambda\). Either 
\begin{itemize}
    \item \(\Dc'\) is given as input \(\ketbra{\Phi^+}\otimes \rho_{0,\lambda}\) or \(\ketbra{\Phi^-}\otimes \rho_{1,\lambda}\), which are valid samples from both mixed states of the form \(\psi_{AB}^\lambda\) or \(\phi_{AB}^\lambda\); or
    \item \(\Dc'\) is given as input \(\ketbra{\Phi^+}\otimes \rho_{1,\lambda}\) or \(\ketbra{\Phi^-}\otimes \rho_{0,\lambda}\), which are only valid samples from the mixed states of the form of \(\psi_{AB}^\lambda\).
\end{itemize}
From here, it is clear that the distinguisher \(\Dc'\) is actually trying to distinguish \(\rho_{0,\lambda}\) and \(\rho_{1,\lambda}\), when the first register is \(\ketbra{\Phi^+}\) or \(\ketbra{\Phi^-}\), but this is actually harder, as half of its possible inputs (\(\ketbra{\Phi^+}\otimes \rho_{0,\lambda}\) and \(\ketbra{\Phi^-}\otimes \rho_{1,\lambda}\)) are valid samples from both mixed states \(\psi_{AB}^\lambda\) and \(\phi_{AB}^\lambda\).
Thus, the advantage of \(\Dc'\) in distinguishing \(\psi_{AB}^\lambda\) and \(\phi_{AB}^\lambda\) is bounded by its capability to distinguish \(\rho_{0,\lambda}\) and \(\rho_{1,\lambda}\), but only given that its input is \(\ketbra{\Phi^+}\otimes \rho_{1,\lambda}\) or \(\ketbra{\Phi^-}\otimes \rho_{0,\lambda}\), otherwise it cannot.\\
Finally, since both states \(\psi_{AB}^\lambda\) and \(\phi_{AB}^\lambda\) may be efficiently preparable locally by any QPT distinguisher, then, by using a hybrid argument, the computational indistinguishability of single states \(\psi_{AB}^\lambda\) and \(\phi_{AB}^\lambda\) is equivalent to the computational indistinguishability of polynomially many copies of these states \({\psi_{AB}^\lambda}^{\otimes p(\lambda)}\) and \({\phi_{AB}^\lambda}^{\otimes p(\lambda)}\), for any polynomial \(p\).

\paragraph{Efficient distillation of the family \(\{\phi_{AB}^\lambda\}_\lambda\).}
In the previous construction, we considered that the distillable entanglement of the family \(\{\phi_{AB}^\lambda\}_\lambda\) does not have to be efficiently accessible.
This is motivated by the fact that this family is not even required to be efficient to prepare, and its existence as a reference frame for the pseudo-entangled states \(\{\psi_{AB}^\lambda\}_\lambda\) is its key aspect.
Nevertheless, we provide a generalization to the setting of the LOCC protocol presented in~\cite{pre:ABV23}, where the parties are given a classical secret key that helps them distill the entanglement efficiently, and allow this key \(k\) to be a quantum state.
Then, we base our procedure in the canonical commitment construction from EFI pairs from~\cite{ITCS:BCQ23}, where we consider the unitary part, \(\hat{\As}_{b,\lambda}\), corresponding to the EFI generator \(\As(1^\lambda,b)\), i.e., \(\hat{\As}_{b,\lambda} \ket{0} = \ket{\chi_{b,\lambda}}_{EK}\), with \({\rho_{b,\lambda}} = \Tr_K(\ketbra{\chi_{b,\lambda}}_{EK}) = \As(1^\lambda,b)\), and let \(k_{b,\lambda} = \Tr_E(\ketbra{\chi_{b,\lambda}}_{EK})\).
Then, efficient distillation may be achieved by distinguishing between the two elements of the mixed state \(\phi_{AB}^\lambda\) (\(\ketbra{\Phi^+}\otimes \rho_{0,\lambda}\) or \(\ketbra{\Phi^-}\otimes \rho_{1,\lambda}\)) by distinguishing \(\rho_{0,\lambda}\) and \(\rho_{1,\lambda}\).
And this may be accomplished given the key \(k_\lambda = k_{0,\lambda}\) (that is kept coherent), by attempting to reconstruct the original state \(\chi_{0,\lambda}' = \ketbra{\chi_{0,\lambda}}\) (from the matching coherent key and EFI pair element).
Clearly, this is only possible if \(b=0\), as if \(b=1\) the resulting state \(\chi_{1,\lambda}'\) is a product state across partition \(E:K\).
Then, by computing \({\hat{\As}_{0,\lambda}}^\dagger\, {\chi_{b,\lambda}'}\,  \hat{\As}_{0,\lambda}\) and measuring all qubits in the computational basis, it is possible to distinguish by checking whether the output is \(\ket{0}\), meaning \(b=0\), or not, meaning \(b=1\).
Thus, the computational distillable entanglement of the family \(\{\phi_{AB}^\lambda\}_\lambda\) given the key \(k_\lambda\) is 1 with error \(\varepsilon\in O(2^{-\lambda})\), 
\begin{equation*}
    \hat{E}_D^{2^{-\lambda}}(\{k_\lambda,\phi_{AB}^\lambda\}_\lambda) = 1.
\end{equation*}

We conclude the overview of our results by analyzing our second main contribution (Corollary~\ref{cor:efi2pe-poly}), which is informally stated below.
\begin{corollary}[EFI pairs $\Rightarrow$ Pseudo-Entanglement with polynomial-gap (Informal)]\label{cor:efi2pe-poly-inf}
    If there exist EFI pairs of states \((\rho_{0,\lambda},\rho_{1,\lambda})\), then there exists pseudo-entanglement with polynomial-gap (in \(\lambda\)).
    I.e., there exist computationally indistinguishable (under polynomially many copies) families of quantum states, \(\{\psi_{AB}^\lambda\}_\lambda\) and \(\{\phi_{AB}^\lambda\}_\lambda\), with \(\hat{E}_C^\varepsilon\left(\{\psi_{AB}^\lambda\}_\lambda\right) = 0\) and \({E}_D^\varepsilon\left(\{\phi_{AB}^\lambda\}_\lambda\right) = q(\lambda)\) (with  polynomial \(q\), and  negligible \(\varepsilon\)).
\end{corollary}

To achieve this result, we first prove a general lemma (formally stated in Lemma~\ref{lemma:amplify}), which may be of independent interest, and states that if there exists a family of pseudo-entangled states with gap \(d - c\), then there also exists a family of pseudo-entangled states with gap \((d - c)q(\lambda)\) for any polynomial \(q\).
This may also be stated in a constructive form, as an amplification of the pseudo-entanglement property of a family of states, meaning that any pseudo-entanglement gap that exists, \(d-c\), may be amplified by any polynomial factor, \((d - c)q(\lambda)\).
The lemma follows from the application of the hybrid argument.
First, one constructs a new pseudo-entangled family by taking copies of the original pseudo-entangled family \(\psi_{AB}^\lambda\), which may be done since it is efficiently preparable using LOCC. 
Accordingly, one also takes copies of the high-entanglement family \(\phi_{AB}^\lambda\), which just needs to exist (it does not need to be efficiently preparable).
Computational indistinguishably holds by the hybrid argument, where single-copy indistinguishability is given by assumption.
Essentially, computational indistinguishability of the base single-copy case is always regarded against non-uniform QPT adversaries, and thus, copies of the \(\phi_{AB}^\lambda\) states, which may not be efficiently preparable, may be given as advice to a distinguisher, so that it is able to build the hybrids with a polynomial number of states, and then use the given many-copies distinguisher to distinguish a single unknown element.

Finally, Corollary~\ref{cor:efi2pe-poly-inf} follows directly by the construction given in Theorem~\ref{thm:efi2pe-inf} that provides pseudo-entanglement with gap-\(1\), together with the construction from the previously described lemma, yielding a construction of pseudo-entanglement with any polynomial gap \(q(\lambda)\).

\section{Background}
Regarding notation, we use the following standard conventions.
We denote by \(\mathbb{N}, \mathbb{R}, \mathbb{C}\) the sets of natural, real, and complex numbers, and use \(\mathcal{H}\) to denote Hilbert spaces.
As usual in quantum information theory, we use the Dirac notation for quantum states.
A density matrix is a positive semi-definite trace-1 matrix that completely characterizes the state of a quantum system.
A pure state is represented by a rank-1 density matrix \(\ketbra{\psi}\) (equivalently, \(\ket{\psi}\)).
A mixed state is represented by a density matrix of rank larger than one, and may be interpreted as a probability distribution over pure states.
The distance between two quantum states \(\rho,\sigma\) is given by the trace distance between their density matrices, 
\begin{equation}
    \TD(\rho,\sigma) = \frac{1}{2}||\rho-\sigma||_1,
\end{equation}
where \(||\rho||_1 = \Tr(\sqrt{\rho^\dagger \rho})\) is the trace norm.
It is also true that, since \(\rho-\sigma\) is Hermitian, the trace distance is also equal to
\begin{equation}
    \TD(\rho,\sigma) = \frac{1}{2} \Tr\left(\left|\rho - \sigma \right|\right) = \frac{1}{2}\sum_{i=1}^r \left|\lambda_i\right|,
\end{equation}
where \(r=\mathrm{rank}(\rho - \sigma)\) and \(\lambda_i\in\R\) are the singular values of \(\rho - \sigma\).
Another measure of distance between quantum states is the fidelity, which is defined as 
\begin{equation}
    \Fid(\rho,\sigma) = \Tr(\sqrt{\sqrt{\rho}\sigma\sqrt{\rho}})^2.
\end{equation}
For a complete reference on quantum computation and quantum information, we refer the reader to~\cite{NC10}.
If \(\As\) is an algorithm, we denote \(y\leftarrow \As(x)\) the output of \(\As\) when run on input \(x\).
If \(\mathcal{X}\) is a probability distribution over some set \(X\), we denote \(x\leftarrow \mathcal{X}\) the sampling of \(x\) from \(X\) according to \(\mathcal{X}\).
If \(X\) is a set, we denote \(x \leftarrow X\) the outcome of sampling an element from \(X\) according to the uniform distribution.
We say an algorithm is QPT if it runs in quantum polynomial time.
If \(x\) and \(y\) are two strings, we denote their concatenation as \(x || y\).

Below, we introduce some definitions and tools that we make use of.
\begin{definition}[Quantum algorithm, generalized quantum circuits, and its unitary part \cite{ITCS:BCQ23}]\label{def:generalizedcircuit}
   Given a discrete universal gate set, a \emph{quantum algorithm} \(\As\) is a sequence of generalized circuits \(\{\As_n\}_{n\in\mathbb{N}}\) over the discrete universal gate set.
    
    An $n$-qubit \emph{generalized quantum circuit} is a sequence of gates from a universal gate set acting on $n$-qubits together with a subset of the qubits that are initialized to the zero state and a second subset of the qubits, not necessarily disjoint, that are traced out after the sequence of gates is applied. 
    
    We say that the unitary operator \(\hat{\As}_n\) is the \emph{unitary part of \(\As_n\)}, if we only take the gates of \(\As_n\) and ignore the input and output qubits that are initialized to zero or traced out.
\end{definition}

\begin{definition}[Negligible function]\label{def:negl}
    A function \(\mu: \mathbb{N}\to \mathbb{R}\) is \emph{negligible}, if for every positive polynomial \(p\), there exists an \(N\in\mathbb{N}\) such that for all integers \(\lambda > N\)
    \begin{equation}
        |\mu(\lambda)| < \frac{1}{p(\lambda)}.
    \end{equation}
\end{definition}

\begin{definition}[Computational indistinguishability]\label{def:computationalind}
    Let \(\lambda\in\mathbb{N}\), and \(\{\rho_\lambda\}_\lambda\) and \(\{\sigma_\lambda\}_\lambda\) be two families of quantum states (analogously defined for probability distributions). 
    \(\{\rho_\lambda\}_\lambda\) and \(\{\sigma_\lambda\}_\lambda\) are said to be \emph{computationally indistinguishable}, if for all non-uniform QPT (i.e., \(\BQPqpolyclass\)) algorithms \(\mathcal{D}\), there exists a negligible function \(\varepsilon:\mathbb{N}\to\mathbb{R}^{\geq 0}\) such that for all \(\lambda\) and all advice states \(\alpha_\lambda\)
    \begin{equation}
        \Adv_\Dc(\rho_\lambda,\sigma_\lambda) = \left| \Pr[1\leftarrow \mathcal{D}^{\alpha_\lambda}(\rho_\lambda)] - \Pr[1\leftarrow \mathcal{D}^{\alpha_\lambda}(\sigma_\lambda)] \right| \leq \varepsilon(\lambda).
    \end{equation}
\end{definition}
\noindent
Note that we may omit the advice state when making a uniform reduction for one non-uniform algorithm, \(\mathcal{A}^{\alpha} \equiv \mathcal{A}\), to another, \(\mathcal{B}^{\beta} \equiv \mathcal{B}\), meaning that the advice state of the constructed algorithm \(\mathcal{B}\) is passed directly to the given algorithm \(\mathcal{A}\), i.e., set \(\alpha = \beta\).
\noindent
\begin{lemma}[\cite{JCSS:GS99}]\label{lemma:hybridargC}
    Computational indistinguishability of two distribution ensembles, \(\{X_n\}_n\) and \(\{Y_n\}_n\), given a single sample is preserved when providing any distinguisher \(\Dc\) with polynomially many samples if either:
    \begin{enumerate}
        \item \(\Dc\) can efficiently sample from the distributions; or
        \item the computational indistinguishability of the ensembles holds for any non-uniform (i.e., \Ppolyclass) distinguisher \(\Dc\).
    \end{enumerate}
\end{lemma}
\noindent
The previous lemma is proved using the well-known hybrid argument, a proof technique that reduces the problem of distinguishing polynomially many samples to distinguishing a single sample, by crafting the input of the many-sample distinguisher and making it able to distinguish a single sample.
This crafting procedure is the fundamental step that requires one of the two conditions above to hold.
It is straightforward to generalize this result to the setting where QPT distinguishers are given polynomially many copies of computationally indistinguishable quantum states.
In particular, that if two families of quantum states \(\{\rho_\lambda\}_\lambda\) and \(\{\sigma_\lambda\}_\lambda\) are computationally indistinguishable against \(\BQPqpolyclass\) under single copies, then they are also computationally indistinguishable under polynomially many copies.
\begin{lemma}\label{lemma:hybridarg}
    Computational indistinguishability against \(\BQPqpolyclass\) (as in Definition~\ref{def:computationalind}) of two families of quantum states, \(\{\rho_\lambda\}_\lambda\) and \(\{\sigma_\lambda\}_\lambda\), given a single copy is preserved when providing any distinguisher \(\Dc\) with polynomially many (\(p(\lambda)\)) identical and independent copies of the states.
\end{lemma}
\begin{proof}
    The proof for this quantum version of the lemma uses the hybrid argument in the same way as the usual classical reduction explained above, where a distinguisher for polynomially (\(p(\lambda)\)) many samples is used with a specially crafted input that allows it to distinguish a single sample.
    But now, the distinguisher is given quantum states as input and advice, instead of samples from a probability distribution.\\
    Let the hybrids 
    \begin{align}
        H_0 &= (\rho_\lambda^1,\rho_\lambda^2,\dots,\rho_\lambda^{p(\lambda)}),\\
        H_i &= (\rho_\lambda^1,\dots,\rho_\lambda^{p(\lambda)-i}, \sigma_\lambda^{p(\lambda)-i+1},\dots,\sigma_\lambda^{p(\lambda)}),\\
        H_{p(\lambda)} &= (\sigma_\lambda^1,\sigma_\lambda^2,\dots,\sigma_\lambda^{p(\lambda)}).
    \end{align}
    For each \(\lambda\), there must exist an \(i\in \{0,1,\dots,p(\lambda)-1\}\), such that 
    \begin{equation}
        i^*_\lambda = \argmax_{i}\Adv_\Dc\left(H_i,H_{i+1}\right),
    \end{equation}
    the ``easiest'' to distinguish among all adjacent hybrids.
    Then, one may bound the advantage of the distinguisher \(\Dc\) by only comparing adjacent hybrids that differ in exactly one position, by expanding in telescoping sum and using the triangle inequality, as
    \begin{align} \begin{split}\label{eq:besthybrid}
        \Adv_\Dc\left(H_0,H_{p(\lambda)}\right) &= \left| \Pr[1\leftarrow\Dc(H_0)] - \Pr[1\leftarrow\Dc(H_{p(\lambda)})]\right| \\
        &= \left| \sum_{i=0}^{p(\lambda)-1} \Pr[1\leftarrow\Dc(H_i)] - \Pr[1\leftarrow\Dc(H_{i+1})] \right|\\
        &\leq \sum_{i=0}^{p(\lambda)-1} \left| \Pr[1\leftarrow\Dc(H_i)] - \Pr[1\leftarrow\Dc(H_{i+1})] \right|\\
        &\leq \sum_{i=0}^{p(\lambda)-1} \Adv_\Dc\left(H_i, H_{i+1}\right)\\
        &\leq \sum_{i=0}^{p(\lambda)-1} \Adv_\Dc\left(H_{i_\lambda^*}, H_{i_\lambda^*+1}\right)\\
        &\leq p(\lambda) \Adv_\Dc\left(H_{i_\lambda^*}, H_{i_\lambda^*+1}\right).
    \end{split} \end{align}
    Now, one may construct a non-uniform distinguisher \(\Dc'^{\alpha_\lambda}(\eta_\lambda)\) that receives as input an unknown state \(\eta_\lambda=\rho_\lambda\) or \(\eta_\lambda=\sigma_\lambda\) and as advice the (polynomially bounded) states 
    \begin{equation}
        \alpha_\lambda = \bigotimes_{i=1}^{i_\lambda^*} \rho_\lambda^i \, \bigotimes_{i=i_\lambda^*+2}^{p(\lambda)} \sigma_\lambda^{i},
    \end{equation}
    and runs \(\Dc\) and returns its output \(d\) as
    \begin{equation}
        d\leftarrow \Dc\left(\bigotimes_{i=1}^{i_\lambda^*} \rho_\lambda^i \bigotimes \eta\, \bigotimes_{i=i_\lambda^*+2}^{p(\lambda)} \sigma_\lambda^{i}\right).
    \end{equation}
    Now, note that 
    \begin{equation}
        \Adv_\Dc\left(H_{i_\lambda^*},H_{i_\lambda^*+1}\right) = \Adv_{{\Dc'}^{\alpha(\lambda)}}\left(\rho_\lambda,\sigma_\lambda\right),
    \end{equation}
    because if \(\eta_\lambda = \sigma_\lambda\), then \(\Dc\) is run on input \(H_{i_\lambda^*}\), and if \(\eta_\lambda = \rho_\lambda\), then \(\Dc\) is run on input \(H_{i_\lambda^*+1}\).
    I.e., we have reduced the problem of distinguishing \(\rho_\lambda\) and \(\sigma_\lambda\) to distinguishing \(H_{i_\lambda^*}\) and \(H_{i_\lambda^*+1}\).\\
    Finally, by assumption of computational indistinguishability of \(\{\rho_\lambda\}_\lambda\) and \(\{\sigma_\lambda\}_\lambda\) against \(\BQPqpolyclass\), for all \(\lambda\) and all advice states \(\alpha_\lambda\) there exists a negligible function \(\varepsilon\) such that \(\Adv_{{\Dc'}^{\alpha(\lambda)}}\left(\rho_\lambda,\sigma_\lambda\right)\leq \varepsilon(\lambda)\).
    Thus, continuing from Equation~\eqref{eq:besthybrid},
    \begin{align} \begin{split}
        \Adv_{{\Dc}}\left(H_0, H_{p(\lambda)}\right) &\leq p(\lambda) \Adv_{\Dc}\left(H_{i_\lambda^*},H_{i_\lambda^*+1}\right)\\
        &\leq p(\lambda) \Adv_{{\Dc'}^{\alpha(\lambda)}}\left(\rho_\lambda,\sigma_\lambda\right) \\
        &\leq p(\lambda)\, \varepsilon(\lambda),
    \end{split} \end{align}
    with \(p(\lambda)\, \varepsilon(\lambda)\) negligible, since \( \varepsilon\) is negligible and \(p\) is a polynomial.
\end{proof}
\noindent
Note that, since we define computational indistinguishability against \(\BQPqpolyclass\) (Definition~\ref{def:computationalind}), we only consider this case in Lemma~\ref{lemma:hybridarg}.
However, in the same way, one may directly generalize the classical version of the hybrid argument for uniform algorithms (\(\BPPclass\)) to the quantum setting (\(\BQPclass\)) assuming that the quantum states are efficient to prepare.

\subsection{Computational Entanglement}
\begin{definition}[LOCC map~\cite{pre:ABV23}]\label{def:locc}
    Let \(\mathcal{H}_A,\mathcal{H}_B\) be two Hilbert spaces representing the input spaces, and \(\mathcal{H}_{\bar{A}},\mathcal{H}_{\bar{B}}\) two Hilbert spaces representing the output spaces.
    An \emph{LOCC map} between two systems, \(A\) and \(B\), is a quantum channel
    \begin{equation}
        \Gamma: \mathcal{H}_A\otimes \mathcal{H}_B \to \mathcal{H}_{\bar{A}} \otimes \mathcal{H}_{\bar{B}},
    \end{equation}
    where two parties holding systems \(A\) and \(B\) may perform any finite number of local quantum operations and measurements, and may communicate any finite amount of classical information.
    It is defined as a finite sequential repetition of the following steps:
    \begin{enumerate}
        \item a quantum channel acts on system \(A\), which has a quantum and classical output;
        \item a quantum channel acts on system \(B\), which has a quantum and classical output;
        \item the corresponding classical outputs are transmitted from \(A\) to \(B\) and \(B\) to \(A\).
    \end{enumerate}
\end{definition}
\noindent
One may consider the quantum channels acting on \(A\) and \(B\) to be generalized quantum circuits acting on \(n_A+t_A+c\) and \(n_B+t_B+c\) qubits respectively, where \(n_A,n_B\) registers are arbitrary input states, \(t_A,t_B\) are ancillary registers, and \(c\) is a classical shared register.
Then, a family of LOCC maps \(\{\hat{\Gamma}_\lambda\}_{\lambda\in\N}\) is \textit{efficient} if there exists a polynomial function of (the size) \(\lambda\) that upper-bounds the total number of gates of its circuit description.

\medskip

In this work, we will follow the operational formalism of~\cite{pre:ABV23} in the study of entanglement.
First, we introduce two relevant entanglement measures, the one-shot distillable entanglement~\cite{JMP:BD10} and the one-shot entanglement cost~\cite{PRL:BD11}, which are both defined under arbitrary LOCC maps.
These entanglement measures assume that the parties are given a \textit{single copy} of a state from which they must distill the maximum possible number of maximally entangled pairs in the case of the distillable entanglement, or construct a given target state using the minimum number of maximally entangled pairs in the case of the entanglement cost.

\begin{definition}[One-shot distillable entanglement]\label{def:ed-it}
    Let \(\varepsilon \in [0,1]\) be the error, \(\Gamma\) be an LOCC map, \(\Fid\) be the fidelity. The \emph{one-shot distillable entanglement} under LOCC of \(\rho_{AB} \in \mathcal{H}_A\otimes\mathcal{H}_B\) is defined as
    \begin{equation}
        E_D^\varepsilon (\rho_{AB}) = \sup_{m,\Gamma} \left\{m\ |\ \Fid\left(\Gamma(\rho_{AB}),\Phi^{\otimes m}\right) \geq 1 - \varepsilon\right\}.
    \end{equation}
\end{definition}

\begin{definition}[One-shot entanglement cost]\label{def:ec-it}
    Let \(\varepsilon \in [0,1]\) be the error, \(\Gamma\) be an LOCC map, \(\Fid\) be the fidelity.
    The \emph{one-shot entanglement cost} under LOCC of \(\rho_{AB} \in \mathcal{H}_A \otimes\mathcal{H}_B\) is defined as
    \begin{equation}
        E_C^\varepsilon(\rho_{AB}) = \inf_{n,\Gamma} \left\{n\ |\ \Fid\left(\Gamma(\Phi^{\otimes n}),\rho_{AB}\right) \geq 1-\varepsilon \right\}.
    \end{equation}
\end{definition}

\medskip 

The previous two definitions of one-shot distillable entanglement and one-shot entanglement cost quantify entanglement as the fidelity of the resulting state of an arbitrary LOCC map and its target state.
However, one may also study these same entanglement measures in the restricted setting where the LOCC map is constrained to be efficient, leading to the theory of computational entanglement~\cite{pre:ABV23}.
At the base of this theory are two new definitions of computational one-shot distillable entanglement and computational one-shot entanglement cost, which we introduce next.

\begin{definition}[Computational one-shot distillable entanglement~\cite{pre:ABV23}]\label{def:ed}
    Let \(\lambda\in\N\) and \(\varepsilon:\mathbb{N}\to [0,1]\) be the size and error parameters.
    Let \(n_A, n_B:\mathbb{N} \to \mathbb{N}\) be two polynomial functions defining the size of a family of bipartite quantum systems (across \(A:B\))  \(\{\rho_{AB}^\lambda\}_\lambda\), of \(n_A(\lambda)+n_B(\lambda)\) qubits.
    Then, the \textit{computational distillable entanglement} is lower-bounded by the function \(m:\N\to\N\),
    \begin{equation}
        \hat{E}_D^\varepsilon\left(\left\{\rho_{AB}^\lambda\right\}_\lambda\right) \geq m,
    \end{equation}
    if, for all \(\lambda\), there exists an efficient LOCC map \(\hat{\Gamma}^\lambda\) acting on \(\rho_{AB}^\lambda\) that outputs \(2m(\lambda)\)-qubit states with error at most \(\varepsilon(\lambda)\), i.e.,
    \begin{equation}
        \Fid\left(\hat{\Gamma}(\rho_{AB}^\lambda),\Phi^{\otimes m})\right) \geq 1 - \varepsilon(\lambda).
    \end{equation}
\end{definition}

\begin{definition}[Computational one-shot entanglement cost~\cite{pre:ABV23}]\label{def:ec}
    Let \(\lambda\in\N\) and \(\varepsilon:\mathbb{N}\to [0,1]\) be the size and error parameters.
    Let \(m_A, m_B:\mathbb{N} \to \mathbb{N}\) be two polynomial functions defining the size of a family of bipartite quantum systems (across \(A:B\)) \(\{\rho_{AB}^\lambda\}_\lambda\), of \(m_A(\lambda) + m_B(\lambda)\) qubits.
    Then, the \textit{computational entanglement cost} is upper-bounded by the function \(n:\N\to\N\),
    \begin{equation}
        \hat{E}_C^\varepsilon\left(\left\{\rho_{AB}^\lambda\right\}_\lambda\right) \leq n,
    \end{equation}
    if, for all \(\lambda\) there exists an efficient LOCC map \(\hat{\Gamma}^\lambda\) acting on \(n(\lambda)\) EPR pairs, and outputs \(\rho_{AB}^\lambda\) with error at most \(\varepsilon(\lambda)\), i.e.,
    \begin{equation}
        \Fid\left(\hat{\Gamma}^\lambda(\Phi^{\otimes n}),\rho_{AB}^\lambda\right) \geq 1 - \varepsilon(\lambda).
    \end{equation}
\end{definition}

In the previous two definitions, the family of bipartite states \(\{\rho_{AB}^\lambda\}_\lambda\) is indexed by the size of the states as a (polynomial) function of \(\lambda\).
However, there might be the case that there are families of states for which many states have the same size parameter \(\lambda\).
This leads to the introduction to the \textit{uniform} definitions of computational one-shot distillable entanglement and computational one-shot entanglement cost~\cite{pre:ABV23}.
These are defined analogously to Definitions~\ref{def:ed} and~\ref{def:ec}, except that, in this setting, each state of size \(\lambda\) is indexed by a classical key \(k\in\{0,1\}^{\kappa(\lambda)}\) (for some polynomial function \(\kappa:\N\to\N\)).
Then, an LOCC map \(\Gamma(k,\rho^k_{AB})\) acts on the bipartite (over \(AA':BB'\)) state \(\ketbra{k}_{A'}\otimes \rho_{AB}^k\otimes \ketbra{k}_{B'}\).
Hence, the \textit{uniform computational one-shot distillable entanglement} is defined as 
    \begin{equation}
        \hat{E}_D^\varepsilon\left(\left\{k,\rho_{AB}^k\right\}_\lambda\right) \geq m,
    \end{equation}
and the \textit{uniform computational one-shot entanglement cost} is defined as
    \begin{equation}
        \hat{E}_C^\varepsilon\left(\left\{k,\rho_{AB}^k\right\}_\lambda\right) \leq n.
    \end{equation}

\subsection{Pseudo-Entanglement}
Two different definitions of pseudo-entanglement are introduced in~\cite{pre:ABF+22} and~\cite{pre:ABV23}.
These two definitions are incompatible, and their main differences are that \cite{pre:ABF+22} restricts the quantum states to pure states and uses the (information-theoretic) entanglement entropy for all possible partitions of the states; and \cite{pre:ABV23} considers arbitrary mixed states and efficient LOCC protocols for operational quantities such as entanglement cost and distillable entanglement for a specific single bipartition of the states.

\begin{definition}[Pseudo-entanglement \cite{pre:ABF+22}]\label{def:pseudoentanglement}
    Let \(\lambda\in \N\), \(n:\N\to\N\), \(\kappa:\N\to\N\) two polynomially bounded functions. A \emph{pseudo-entangled state ensemble with gap \(f(\lambda)\) vs.\ \(g(\lambda)\)} consists of two ensembles of \(n(\lambda)\)-qubit states \(\{\ket{\psi_k}\}_k,\{\ket{\phi_k}\}_k\), indexed by \(k\in\{0,1\}^{\kappa(\lambda)}\), such that:
    \begin{enumerate}
        \item Given \(k\), \(\ket{\psi_k}\) and \(\ket{\phi_k}\) are efficiently preparable by a uniform, polynomial size quantum circuit.
        \item The entanglement entropy across every cut of \(\ket{\psi_k}\) is \(\Theta(f(\lambda))\), and   the entanglement entropy across every cut of \(\ket{\phi_k}\) is \(\Theta(g(\lambda))\).
        \item \(\rho=\E_k\left[\ketbra{\psi_k}^{\otimes p(\lambda)}\right]\) and \(\sigma =\E_k\left[\ketbra{\phi_k}^{\otimes p(\lambda)}\right]\) are computationally indistinguishable, for any polynomial \(p\).
    \end{enumerate}
\end{definition}

\begin{definition}[Pseudo-entanglement \cite{pre:ABV23}]\label{def:pseudoentanglement2}
    Let \(\lambda\in \mathbb{N}\), \(n:\mathbb{N} \to \mathbb{N}\), \(\kappa:\N\to\N\) two polynomially bounded functions, and \(\varepsilon:\mathbb{N}\to [0,1]\) and \(c,d: \mathbb{N}\to \mathbb{N}\) arbitrary.
    A family of \(n(\lambda)\)-qubit bipartite quantum states \(\{\psi_{AB}^k\}_{k\in\{0,1\}^{\kappa(\lambda)}}\) is said to \emph{have pseudo-entanglement} \((\varepsilon,c,d)\) if there is a family of \(n(\lambda)\)-qubit bipartite quantum states \(\{\phi_{AB}^k\}_{k\in\{0,1\}^{\kappa(\lambda)}}\), such that:
    \begin{enumerate}
        \item The computational entanglement cost is upper-bounded as \(\hat{E}_C^\varepsilon\left(\{k,\psi_{AB}^k\}_k\right) \leq c\).
        \item  The computational distillable entanglement is lower-bounded as \(\hat{E}_D^\varepsilon\left(\{k,\phi_{AB}^k\}_k\right) \geq d\).
        \item \(\rho_{AB} = \mathbb{E}_k\left[{\psi_{AB}^k}^{\otimes p(\lambda)}\right]\) and \(\sigma_{AB} = \mathbb{E}_k\left[{\phi_{AB}^k}^{\otimes p(\lambda)}\right]\) are computationally indistinguishable, for any polynomial \(p\).
    \end{enumerate}
\end{definition}

We consider both formalisms and analyze how the existence of EFI pairs relate to the existence of pseudo-entangled states when considering both definitions.
For that purpose, we make some adaptations, depending on the setting.
Namely, we consider the following relaxations to the previous definitions (Definition~\ref{def:pseudoentanglement} and~\ref{def:pseudoentanglement2}):
\begin{itemize}
\item Regarding Definition~\ref{def:pseudoentanglement}, we not consider the keys $k$ required to efficiently prepare the states as we use a generator of EFI pairs, which are indistinguishable even given the description of the circuit.
With this consideration, our approach may be seen as a kind of public-key pseudo-entanglement construction.
Also, we consider pseudo-entanglement across a single partition (as in \cite{pre:GH20}) instead of for all partitions (of large enough size to avoid trivial distinguishers).
This relaxed definition seems to still be useful for two-party protocols (as in Definition~\ref{def:pseudoentanglement2}, where only one bipartition is considered).
Another relaxation is that we must consider that the pseudo-entangled states may be mixed states, since there is the possibility that the EFI pair elements might be mixed states.
This implies that we must consider some appropriate entanglement measure $E$ to quantify the entanglement of mixed states.
To be as general as possible, we quantify the entanglement of a state $\rho_{AB}$ across partition $A:B$ using any entanglement measure $E$ that fulfills
\begin{equation}
    E_D^\varepsilon(\rho_{AB})\leq E^\varepsilon(\rho_{AB})\leq E_C^\varepsilon(\rho_{AB}).
\end{equation}
Here, we consider \(E_D\) as \(E_C\) as extreme
measures, since this guarantees that $E$ has many properties expected from a measure of entanglement~\cite{JMP:DHR02}, and \(\varepsilon\) for $E$ may be thought as some approximation error, as in \(E_D^\varepsilon\) (Definition~\ref{def:ed-it}) and \(E_C^\varepsilon\) (Definition~\ref{def:ec-it}).
Finally, we do not require the states of the high-entanglement reference family to be efficiently preparable, which matches the case for Definition~\ref{def:pseudoentanglement2}.

\item Regarding Definition~\ref{def:pseudoentanglement2}, we do not consider keys $k$ (nor the uniform computational measures) as they are not used in our constructions, and we do not require the distillation of the high-entanglement reference family to be efficient.
This is motivated by the fact that this family may not even be efficient to prepare, and its existence as a reference is the relevant concept as it shows that a ``cheaper'' low-entanglement family of states \(\{\psi_{AB}^\lambda\}_\lambda\) may be considered instead, without being noticed by any QPT algorithm.
\end{itemize}

We present the adapted definitions from Definition~\ref{def:pseudoentanglement} and~\ref{def:pseudoentanglement2} according to the above reasoning in Definition~\ref{def:pseudoentanglement3} and~\ref{def:pseudoentanglement4}.

\begin{definition}[Pseudo-entanglement \cite{pre:ABF+22} of mixed states across a single cut]\label{def:pseudoentanglement3}
    Let \(\lambda\in \N\), \(n:\mathbb{N} \to \mathbb{N}\) a polynomially bounded function. A \emph{pseudo-entangled state ensemble with gap \(f(\lambda)\) vs.\ \(g(\lambda)\)} consists of two ensembles of \(n(\lambda)\)-qubit states \(\{\psi_{AB}^\lambda\}_\lambda,\{\phi_{AB}^\lambda\}_\lambda\), such that:
    \begin{enumerate}
        \item \(\psi_{AB}^\lambda\) is efficiently preparable by a uniform, polynomial size quantum circuit.
        \item For an entanglement measure $E$, such that $E_D^\varepsilon(\rho_{AB})\leq E^\varepsilon(\rho_{AB})\leq E_C^\varepsilon(\rho_{AB})$, across a single cut ($A:B$), \(E^\varepsilon(\psi_{AB}^\lambda) \in \Theta(f(\lambda))\) and \(E^\varepsilon(\phi_{AB}^\lambda) \in \Theta(g(\lambda))\).
        \item \(\left\{{\psi_{AB}^\lambda}^{\otimes p(\lambda)}\right\}_\lambda\) and \(\left\{{\phi_{AB}^\lambda}^{\otimes p(\lambda)}\right\}_\lambda\) are computationally indistinguishable, for any polynomial \(p\).
    \end{enumerate}
\end{definition}

\begin{definition}[Pseudo-entanglement~\cite{pre:ABV23} without efficient distillation]\label{def:pseudoentanglement4}
    Let \(\lambda\in \mathbb{N}\), \(n:\mathbb{N} \to \mathbb{N}\) a polynomially bounded function, and \(\varepsilon:\mathbb{N}\to [0,1]\) and \(c,d: \mathbb{N}\to \mathbb{N}\) arbitrary.
    A family of \(n(\lambda)\)-qubit bipartite quantum states \(\{\psi_{AB}^\lambda\}_\lambda\) is said to \emph{have pseudo-entanglement} \((\varepsilon,c,d)\) if there is a family of \(n(\lambda)\)-qubit bipartite quantum states \(\{\phi_{AB}^\lambda\}_\lambda\), such that:
    \begin{enumerate}
        \item The computational entanglement cost is upper-bounded as \(\hat{E}_C^\varepsilon\left(\{\psi_{AB}^\lambda\}_\lambda\right) \leq c\).
        \item  The distillable entanglement is lower-bounded as \(E_D^\varepsilon\left(\{\phi_{AB}^\lambda\}_\lambda\right) \geq d\).
        \item \(\left\{{\psi_{AB}^\lambda}^{\otimes p(\lambda)}\right\}_\lambda\) and \(\left\{{\phi_{AB}^\lambda}^{\otimes p(\lambda)}\right\}_\lambda\) are computationally indistinguishable, for any polynomial \(p\).
    \end{enumerate}
\end{definition}

\subsection{EFI Pairs}
\begin{definition}[EFI pairs \cite{ITCS:BCQ23}]\label{def:efi}
    Let \(\lambda\in \mathbb{N}\).
    The tuple \((\rho_{0,\lambda},\rho_{1,\lambda})\) is a \emph{pair of EFI states} if it fulfills the following properties:
    \begin{enumerate}
        \item \emph{Efficient generation}: There exists a uniform QPT algorithm that on input \((1^\lambda,b)\) for \(b\in \{0,1\}\) outputs the state \(\rho_{b,\lambda}\).
        \item \emph{Statistically far}: \(\TD(\rho_{0,\lambda},\rho_{1,\lambda}) \geq \frac{1}{p(\lambda)}\), for any polynomial \(p\).
        \item \emph{Computational indistinguishability}: \(\{\rho_{0,\lambda}\}_\lambda\) and \(\{\rho_{1,\lambda}\}_\lambda\) are computationally indistinguishable.
    \end{enumerate}
\end{definition}

\begin{lemma}[\cite{ITCS:BCQ23}]\label{lemma:expefi}
    Let \(\rho_{0,\lambda},\rho_{1,\lambda}\) be an EFI pair with trace distance larger than any polynomial \(p\), i.e., \(\TD(\rho_{0,\lambda},\rho_{1,\lambda}) \geq 1/p(\lambda)\).
    One can create another EFI pair \(\rho_{0,\lambda}'={\rho_{0,\lambda}}^{\otimes q(\lambda)},\rho_{1,\lambda}'={\rho_{1,\lambda}}^{\otimes q(\lambda)}\) (for polynomial \(q\)) with trace distance exponentially close to \(1\), i.e.,
    \begin{equation}
        \TD\left(\rho_{0,\lambda}',\rho_{1,\lambda}'\right) \geq 1-2^{-\lambda},
    \end{equation}
    while maintaining efficient generation and computational indistinguishability (by a hybrid argument), since
    \begin{equation}
        \TD\left(\rho^{\otimes n}, \sigma^{\otimes n}\right) \geq 1-e^\frac{-n\TD(\rho,\sigma)}{2},
    \end{equation}
    and \(n = q(\lambda) = 2\,\lambda\, p(\lambda)\).
\end{lemma}

\section{EFI Pairs Imply Pseudo-Entanglement}\label{sec:efi2pe}
In this section, we present our main result.
We show that if EFI pairs of quantum states (as defined in Definition~\ref{def:efi}) exist, then families of pseudo-entangled states (as defined in Definition~\ref{def:pseudoentanglement3} and~\ref{def:pseudoentanglement4}) must also exist.
For this, we construct two families of pseudo-entangled states, \(\{\psi_{AB}^\lambda\}_\lambda\) and \(\{\phi_{AB}^\lambda\}_\lambda\), from EFI pairs, and show that they fulfill the required properties of each of the definitions.
I.e., for Definition~\ref{def:pseudoentanglement3}, show that elements of the first family are efficiently preparable and have low entanglement, \(E^\varepsilon(\{\psi_{AB}^\lambda\}_\lambda) \in \Theta(f(\lambda))\), while they are pseudo-entangled since they are computationally indistinguishable from elements of another family with higher entanglement, \(E^\varepsilon(\{\phi_{AB}^\lambda\}_\lambda) \in \Theta(g(\lambda))\).
And, for Definition~\ref{def:pseudoentanglement4}, show that elements of the first family are efficiently preparable, \(\{\psi_{AB}^\lambda\}_\lambda\), with \(\hat{E}_C^\varepsilon(\{\psi_{AB}^\lambda\}_\lambda)\leq c\), and are pseudo-entangled since they are computationally indistinguishable from elements of another family with higher entanglement, \(\{\phi_{AB}^\lambda\}_\lambda\), with \(E_D^\varepsilon(\{\phi_{AB}^\lambda\}_\lambda)\geq d\).

In particular, for parameter \(\lambda\in\N\), we show this implication for a pseudo-entanglement separation of gap $1$ (constant in \(\lambda\)).
We construct the first family, \(\{\psi_{AB}^\lambda\}\), such that \(E_D^\varepsilon(\{\psi_{AB}^\lambda\}_\lambda) = E_C^\varepsilon(\{\psi_{AB}^\lambda\}_\lambda) = \hat{E}_C^\varepsilon(\{\psi_{AB}^\lambda\}_\lambda) = 0\) (i.e., the states \(\psi_{AB}^\lambda\) may be prepared from LOCC by two parties that share no entanglement, since these are separable states), and show that this family is indistinguishable under polynomially many copies from a second family, \(\{\phi_{AB}^\lambda\}_\lambda\), such that  \(E_D^\varepsilon(\{\phi_{AB}^\lambda\}_\lambda) = E_C^\varepsilon(\{\phi_{AB}^\lambda\}_\lambda) = 1\) (i.e., the parties are able to distill one EPR pair from the entangled states \(\phi_{AB}^\lambda\)).
This result is formally stated and proved in Theorem~\ref{thm:efi2pe}.

\begin{theorem}\label{thm:efi2pe}
    Let \(\lambda\in\mathbb{N}\), \(n:\mathbb{N}\to\mathbb{N}\) a polynomially bounded function, and \(\As\) a QPT algorithm generating an EFI pair as defined in Definition~\ref{def:efi}.
    There exist two families of \((n(\lambda)+1)\)-qubit bipartite quantum states, \(\{\psi_{AB}^\lambda\}_\lambda\) with \(\hat{E}_C^\varepsilon\left(\{\psi_{AB}^\lambda\}_\lambda\right) = 0\) (also for any \(E^\varepsilon\) such that \(E^\varepsilon\leq E_C^\varepsilon\leq \hat{E}_C^\varepsilon\)), and \(\{\phi_{AB}^\lambda\}_\lambda\) with \(E_D^\varepsilon\left(\{\phi_{AB}^\lambda\}_\lambda\right) \geq 1\) (also for any \(E^\varepsilon\) such that \(E_D^\varepsilon\leq E^\varepsilon\)), with \(\varepsilon \in O(2^{-\lambda})\), such that \(\left\{{\psi_{AB}^\lambda}^{\otimes p(\lambda)}\right\}_\lambda\) and \(\left\{{\phi_{AB}^\lambda}^{\otimes p(\lambda)}\right\}_\lambda\) are computationally indistinguishable, for any polynomial \(p\).
    I.e., \(\{\psi_{AB}^\lambda\}_\lambda\) is pseudo-entangled according to both Definitions~\ref{def:pseudoentanglement3} and~\ref{def:pseudoentanglement4}.
\end{theorem}

\begin{proof}
To prove Theorem~\ref{thm:efi2pe}, we give a reduction from the existence of EFI pairs to the existence of families of pseudo-entangled states.
Then, we show that the families verify the required properties to exhibit pseudo-entanglement. 
Namely, the separation of the entanglement of the families together with their computational indistinguishability given polynomially many copies.

\paragraph{Construction:}
    First, we construct the pseudo-entangled family \(\{\psi_{AB}^\lambda\}_\lambda\) and the reference family \(\{\phi_{AB}^\lambda\}_\lambda\), given an algorithm that generates EFI pairs, \(\rho_{b,\lambda}\leftarrow \As(1^\lambda,b)\).
    From Lemma~\ref{lemma:expefi}, in this proof, we consider that the two states of the EFI pair (\(\rho_{0,\lambda},\rho_{1,\lambda}\)) have trace distance exponentially close to \(1\).
    Let \(\lambda \in \mathbb{N}\).\\
    Construct \(\psi_{AB}^\lambda\) as
    \begin{align}\label{eq:psi}
        \psi_{AB}^\lambda &= \frac{1}{4}\left(\ketbra{\Phi^+}_{AB} + \ketbra{\Phi^-}_{AB}\right)\otimes \left({\rho_{0,\lambda}}_A + {\rho_{1,\lambda}}_A\right).
    \end{align}
    Construct  \(\phi_{AB}^\lambda\) as 
    \begin{align}\label{eq:phi}
        \phi_{AB}^\lambda &= \frac{1}{2}\left(\ketbra{\Phi^+}_{AB}\otimes {\rho_{0,\lambda}}_A + \ketbra{\Phi^-}_{AB}\otimes {\rho_{1,\lambda}}_A \right).
    \end{align}
    
    \paragraph{Entanglement separation:}
    On one hand, the states \(\psi_{AB}^\lambda\) are separable, thus, they may be constructed using LOCC.
    First, party \(A\) flips a fair coin: if it gets heads, it sets the second register to \(\rho_{0,\lambda}\leftarrow \As(1^\lambda,0)\); if it gets tails, it sets the second register to \(\rho_{1,\lambda}\leftarrow \As(1^\lambda,1)\).
    Second, party \(A\) flips another (independent) fair coin: if it gets heads, it prepares state \(\ketbra{0}_A\) and tells the other party to prepare \(\ketbra{0}_B\); if it gets tails, it prepares state \(\ketbra{1}_A\) and tells the other party to prepare \(\ketbra{1}_B\).
    The resulting state is
    \begin{align}\begin{split}
        \frac{1}{4}(\ketbra{00}_{AB} + \ketbra{11}_{AB}) \otimes \left(\rho_{0,\lambda} + \rho_{1,\lambda}\right)_A &= \frac{1}{4}\left(\ketbra{\Phi^+}_{AB} + \ketbra{\Phi^-}_{AB}\right) \otimes \left({\rho_{0,\lambda}}_A + {\rho_{1,\lambda}}_A\right) \\
        &= \psi_{AB}^\lambda.
    \end{split}\end{align}
    So,
    \begin{equation}
        \hat{E}_C^0\left(\{\psi_{AB}^\lambda\}_\lambda\right)=0.
    \end{equation}
    On the other hand, the states \(\phi_{AB}^\lambda\) are entangled.
    However, in order to access the entanglement of the states \(\phi_{AB}^\lambda\), it is necessary to distinguish whether the element of the mixed state that the parties hold is \(\ketbra{\Phi^+}_{AB}\otimes {\rho_{0,\lambda}}_A\) or \(\ketbra{\Phi^-}_{AB}\otimes {\rho_{1,\lambda}}_A\).
    First, note that the construction is asymmetric, and only party \(A\) has access to a register with an EFI pair element (with \(1/2\) probability of being each possibility, \(\rho_{0,\lambda}\) or \(\rho_{1,\lambda}\)).
    Notwithstanding, the construction (for both states \(\phi_{AB}^\lambda\) and \(\psi_{AB}^\lambda\)) may be easily generalized such that each party would have an independent EFI pair with corresponding elements (i.e., same bit \(b\), \({\rho_{b,\lambda}}_A \otimes {\rho_{b,\lambda}}_B\)).
    However, since it is not necessary, we do not consider this possibility.
    That is the case because since the parties are in an LOCC scenario, if \(A\) can distinguish the two cases, it can just communicate classically to \(B\) if they actually have \(\ket{\Phi^+}_{AB}\) or \(\ket{\Phi^-}_{AB}\).
    Moreover, we do not require that the parties executing the distillation protocol are efficient, so, they can perform non-efficient operations, which is required for \(A\) to be able to distinguish the two possible cases and know which EPR pair the parties share.
    Now, recall that, by Lemma~\ref{lemma:expefi}, the trace distance between \(\rho_{0,\lambda}\) and \(\rho_{1,\lambda}\) may be considered to be exponentially close to \(1\), \(\TD\left({\rho_{0,\lambda}}_A, {\rho_{1,\lambda}}_A\right) \geq 1- 2^{-\lambda}\).
    We proceed to show that this implies that there exists an algorithm \(\Dc\) (not necessarily QPT) that distinguishes \(\rho_{0,\lambda}\) and \(\rho_{1,\lambda}\) with overwhelming probability in \(\lambda\). I.e., for all \(\lambda\),
    \begin{equation}
        \TD(\rho_{0,\lambda},\rho_{1,\lambda}) \geq 1-2^{-\lambda} \quad \Rightarrow \quad \exists \Dc \quad \Adv_\Dc (\rho_{0,\lambda},\rho_{1,\lambda}) \geq 1-2^{-\lambda}.
    \end{equation}
    The argument follows from describing the advantage in terms of the trace distance, and then decomposing the spectrum of the difference of the two states, such that one can construct a POVM that maximizes the guessing probability of a distinguisher. (For a detailed explanation, see, e.g.,~\cite{W13}.)
    First, let \(\{E_0,E_1\}\) be an arbitrary two element POVM.
    Then, for each \(E_b\), there exists a corresponding algorithm \(\Dc\), such that for \(\rho_{0,\lambda}\), \(\Pr[b\leftarrow \Dc(\rho_{0,\lambda})] = \Tr(E_b\, \rho_{0,\lambda})\), and for \(\rho_{1,\lambda}\), \(\Pr[b\leftarrow \Dc(\rho_{1,\lambda})] = \Tr(E_b\, \rho_{1,\lambda})\).
    Thus,
    \begin{align}\begin{split}\label{eq:adv2tr}
        \Adv_\Dc(\rho_{0,\lambda},\rho_{1,\lambda}) &= \left| \Pr[1\leftarrow \Dc(\rho_{0,\lambda})] - \Pr[1\leftarrow \Dc(\rho_{1,\lambda})] \right|\\
        &=\left| \Tr\left(E_1\,\rho_{0,\lambda} \right) - \Tr\left(E_1\,\rho_{1,\lambda}\right) \right|\\
        &=\left| \Tr\left(E_1  \left(\rho_{0,\lambda} - \rho_{1,\lambda}\right)\right) \right|.
    \end{split}\end{align}
    Then, since we only need that there exists some \(\Dc\), we may maximize over all possible \(E_1\) (essentially choosing the best distinguisher \(\Dc\)).
    Moreover, \(\rho_{0,\lambda}-\rho_{1,\lambda}\) is Hermitian, so it may be diagonalized as
    \begin{equation}
        \rho_{0,\lambda}-\rho_{1,\lambda} = \sum_i \lambda_i\ketbra{e_i}{e_i},
    \end{equation}
    with \(\{\lambda_i\}_i\) the set of real eigenvalues and \(\{\ket{e_i}\}_i\) the set of eigenvectors, which form an orthonormal basis.
    One may then rewrite the Hermitian operator \(\rho_{0,\lambda}-\rho_{1,\lambda}\) by splitting its positive and negative part of the spectrum,
    \begin{align}\begin{split}\label{eq:splitPQ}
        \rho_{0,\lambda}-\rho_{1,\lambda} &= \sum_{i:\lambda_i\geq 0} \lambda_i\ketbra{e_i} + \sum_{i:\lambda_i<0} \lambda_i\ketbra{e_i}\\
        &=\sum_{i:\lambda_i\geq 0} \lambda_i\ketbra{e_i}  - \sum_{i:\lambda_i<0} |\lambda_i|\ketbra{e_i}\\
        &= P - Q.
    \end{split}\end{align}
    Furthermore, since \(P\) and \(Q\) are orthogonal,
    \begin{equation}
        |\rho_{0,\lambda}-\rho_{1,\lambda}| = |P-Q| = P+Q.
    \end{equation}
    So, 
    \begin{align}\begin{split}\label{eq:norm1trP}
        ||\rho_{0,\lambda}-\rho_{1,\lambda}||_1 = \Tr(|\rho_{0,\lambda}-\rho_{1,\lambda}|) = \Tr(P+Q) = 2\Tr(P),
    \end{split}\end{align}
        since \(\Tr(\rho_{0,\lambda}-\rho_{1,\lambda}) = \Tr(\rho_{0,\lambda}) - \Tr(\rho_{1,\lambda}) = 0 \Rightarrow \Tr(P) = \Tr(Q)\).\\[1em]
    Therefore, one can let
    \begin{equation}\label{eq:E2Pi}
        E_1 = \Pi_P = \sum_{i:\lambda_i\geq 0} \ketbra{e_i}
    \end{equation}
    be a projector onto the positive subspace of the eigenspace, such that (continuing from Equation~\eqref{eq:adv2tr})
    \begin{equation}\begin{alignedat}{2}
        \Adv_\Dc(\rho_{0,\lambda},\rho_{1,\lambda}) &= \left|\Tr(\Pi_P\,(\rho_{0,\lambda}-\rho_{1,\lambda}))\right| &\qquad\qquad& \text{(from Equation \eqref{eq:adv2tr} and \eqref{eq:E2Pi})}\\
        &= \left|\Tr(\Pi_P\,(P-Q))\right| && \text{(from Equation~\eqref{eq:splitPQ})}\\
        &= \left|\Tr(P)\right| && \text{(since \(\Pi_P\,P = P\) and \(\Pi_P\,Q = 0\))} \\
        &= \frac{1}{2}||\rho_{0,\lambda}-\rho_{1,\lambda}||_1 && \text{(from Equation~\eqref{eq:norm1trP})}\\
        &= \TD(\rho_{0,\lambda},\rho_{1,\lambda}) && \text{(by definition)}\\
        &\geq 1 - 2^{-{\lambda}}. && \text{(by assumption)}
    \end{alignedat}\end{equation}
    Finally, this means that there exists an algorithm \(\Dc\) that allows the party holding system \(A\) to run \(\Dc\) and distinguish the case where the EFI pair element is \(\rho_{0,\lambda}\) from the case where it is  \(\rho_{1,\lambda}\).
    Therefore, it learns the state that was sampled in the mixed state \(\phi_{AB}^\lambda\), and thus the EPR pair shared with \(B\), and may communicate (classically) this information to \(B\), meaning \(A\) and \(B\) are able to distill either \(\ket{\Phi^+}_{AB}\) or \(\ket{\Phi^-}_{AB}\).
    Thus, \({E}_D^\varepsilon\left(\{\phi_{AB}^\lambda\}_\lambda\right)=1\), as desired, and \(\varepsilon \in O(2^{-\lambda})\) is negligible in \(\lambda\), i.e., the probability to guess the wrong EFI pair element (and so guess the wrong EPR pair they hold) is exponentially small in \(\lambda\). Thus,
    \begin{equation}
        {E}_D^{2^{-\lambda}}\left(\{\phi_{AB}^\lambda\}_\lambda\right)=1.
    \end{equation}
    
    \paragraph{Computational indistinguishability:}
    By definition of EFI pairs, \(\rho_{0,\lambda}\) and \(\rho_{1,\lambda}\) are computationally indistinguishable, i.e.\ (Definition~\ref{def:computationalind}),
    \begin{equation}
        \forall \Dc\ \exists \varepsilon\quad \Adv_\Dc(\rho_{0,\lambda}, \rho_{1,\lambda}) = \left|\Pr[1\leftarrow \Dc({\rho_{0,\lambda}})] - \Pr[1\leftarrow \Dc(\rho_{1,\lambda})]\right| \leq \varepsilon(\lambda),
    \end{equation}
    where \(\Dc\) are non-uniform QPT algorithms and \(\varepsilon\) is a negligible function.
    We show that the computational indistinguishability of \(\rho_{0,\lambda}\) and \(\rho_{1,\lambda}\) implies the computational indistinguishability of \(\psi_{AB}^\lambda\) and \(\phi_{AB}^\lambda\), i.e.,
    \begin{equation}
        \forall \Dc\ \exists \varepsilon\quad \Adv_\Dc (\rho_{0,\lambda}, \rho_{1,\lambda}) \leq \varepsilon(\lambda) \quad \Rightarrow \quad \forall \Dc'\ \exists \varepsilon'\quad \Adv_{\Dc'} (\psi_{AB}^\lambda,\phi_{AB}^\lambda) \leq \varepsilon'(\lambda),
    \end{equation}
    where \(\Dc,\Dc'\) are non-uniform QPT algorithms and \(\varepsilon,\varepsilon'\) are negligible functions.\\
    By definition, the advantage of an algorithm \(\Dc'\) in distinguishing \(\psi_{AB}^\lambda\) and \(\phi_{AB}^\lambda\) is
    \begin{align}
        \Adv_{\Dc'}(\psi_{AB}^\lambda,\phi_{AB}^\lambda) = \left|\Pr[1\leftarrow {\Dc'}(\psi_{AB}^\lambda)] - \Pr[1\leftarrow {\Dc'}(\phi_{AB}^\lambda)]\right|.
    \end{align}
    Expanding all the possibilities of the states \(\Phi\otimes \rho_{b,\lambda}\), for \(\Phi \in \{\ketbra{\Phi^+},\ketbra{\Phi^-}\}\) and \(b\in\{0,1\}\), from either mixed state \(\psi_{AB}^\lambda\) and \(\phi_{AB}^\lambda\) (i.e., all samples for either distribution),
    \begin{align}\begin{split} 
        \Adv_{\Dc'}(\psi_{AB}^\lambda,\phi_{AB}^\lambda) &= \left|\Pr[1\leftarrow {\Dc'}(\psi_{AB}^\lambda)] - \Pr[1\leftarrow {\Dc'}(\phi_{AB}^\lambda)]\right| \\
        &= \left|\Pr[1\leftarrow {\Dc'}(\psi_{AB}^\lambda) \;\bigg|\; \psi_{AB}^\lambda = \ketbra{\Phi^+}\otimes \rho_{0,\lambda}] \Pr[\psi_{AB}^\lambda = \ketbra{\Phi^+}\otimes \rho_{0,\lambda}] \right.\\ 
        &\qquad + \Pr[1\leftarrow {\Dc'}(\psi_{AB}^\lambda) \;\bigg|\; \psi_{AB}^\lambda = \ketbra{\Phi^+}\otimes \rho_{1,\lambda}]\Pr[\psi_{AB}^\lambda = \ketbra{\Phi^+}\otimes \rho_{1,\lambda}] \\
        &\qquad + \Pr[1\leftarrow {\Dc'}(\psi_{AB}^\lambda) \;\bigg|\; \psi_{AB}^\lambda = \ketbra{\Phi^-}\otimes \rho_{0,\lambda}] \Pr[\psi_{AB}^\lambda = \ketbra{\Phi^-}\otimes \rho_{0,\lambda}]\\
        &\qquad + \Pr[1\leftarrow {\Dc'}(\psi_{AB}^\lambda) \;\bigg|\; \psi_{AB}^\lambda = \ketbra{\Phi^-}\otimes \rho_{1,\lambda}]\Pr[\psi_{AB}^\lambda = \ketbra{\Phi^-}\otimes \rho_{1,\lambda}] \\
        &\qquad - \Pr[1\leftarrow {\Dc'}(\phi_{AB}^\lambda) \;\bigg|\; \phi_{AB}^\lambda = \ketbra{\Phi^+}\otimes \rho_{0,\lambda}]\Pr[\phi_{AB}^\lambda = \ketbra{\Phi^+}\otimes \rho_{0,\lambda}]\\
        &\qquad - \left. \Pr[1\leftarrow {\Dc'}(\phi_{AB}^\lambda) \;\bigg|\; \phi_{AB}^\lambda = \ketbra{\Phi^-}\otimes \rho_{1,\lambda}]\Pr[\phi_{AB}^\lambda = \ketbra{\Phi^-}\otimes \rho_{1,\lambda}]\right|.
    \end{split}\end{align}
    So, one may substitute the conditional probability as
    \begin{align}\begin{split} 
        \Adv_{\Dc'}(\psi_{AB}^\lambda,\phi_{AB}^\lambda) &= \left|\Pr[1\leftarrow {\Dc'}\left(\ketbra{\Phi^+}\otimes \rho_{0,\lambda}\right)] \Pr[\psi_{AB}^\lambda = \ketbra{\Phi^+}\otimes \rho_{0,\lambda}] \right.\\ 
        &\qquad + \Pr[1\leftarrow {\Dc'}\left(\ketbra{\Phi^+}\otimes \rho_{1,\lambda}\right)]\Pr[\psi_{AB}^\lambda = \ketbra{\Phi^+}\otimes \rho_{1,\lambda}] \\
        &\qquad + \Pr[1\leftarrow {\Dc'}\left(\ketbra{\Phi^-}\otimes \rho_{0,\lambda}\right)] \Pr[\psi_{AB}^\lambda = \ketbra{\Phi^-}\otimes \rho_{0,\lambda}]\\
        &\qquad + \Pr[1\leftarrow {\Dc'}\left(\ketbra{\Phi^-}\otimes \rho_{1,\lambda}\right)]\Pr[\psi_{AB}^\lambda = \ketbra{\Phi^-}\otimes \rho_{1,\lambda}] \\
        &\qquad - \Pr[1\leftarrow {\Dc'}\left(\ketbra{\Phi^+}\otimes \rho_{0,\lambda}\right)]\Pr[\phi_{AB}^\lambda = \ketbra{\Phi^+}\otimes \rho_{0,\lambda}]\\
        &\qquad - \left. \Pr[1\leftarrow {\Dc'}\left(\ketbra{\Phi^-}\otimes \rho_{1,\lambda}\right)] \Pr[\phi_{AB}^\lambda = \ketbra{\Phi^-}\otimes \rho_{1,\lambda}]\right|.
    \end{split}\end{align}
    Moreover, the priors for each possible sample from the mixed states are given by the construction of the states \(\psi_{AB}^\lambda\) (Equation~\eqref{eq:psi}) and \(\phi_{AB}^\lambda\) (Equation~\eqref{eq:phi}), hence
    \begin{align}\begin{split} 
        \Adv_{\Dc'}(\psi_{AB}^\lambda,\phi_{AB}^\lambda) &= \frac{1}{4} \left|\Pr[1\leftarrow {\Dc'}\left(\ketbra{\Phi^+}\otimes \rho_{0,\lambda}\right)] + \Pr[1\leftarrow {\Dc'}\left(\ketbra{\Phi^+}\otimes \rho_{1,\lambda}\right)] \right. \\
        &\qquad + \Pr[1\leftarrow {\Dc'}\left(\ketbra{\Phi^-}\otimes \rho_{0,\lambda}\right)] + \Pr[1\leftarrow {\Dc'}\left(\ketbra{\Phi^-}\otimes \rho_{1,\lambda}\right)]\\
        &\qquad - \left. 2\Pr[1\leftarrow {\Dc'}\left(\ketbra{\Phi^+}\otimes \rho_{0,\lambda}\right)] - 2 \Pr[1\leftarrow {\Dc'}\left(\ketbra{\Phi^-}\otimes \rho_{1,\lambda}\right)] \right|\\
        &= \frac{1}{4} \left|-\Pr[1\leftarrow {\Dc'}\left(\ketbra{\Phi^+}\otimes \rho_{0,\lambda}\right)] + \Pr[1\leftarrow {\Dc'}\left(\ketbra{\Phi^+}\otimes \rho_{1,\lambda}\right)] \right. \\
        &\qquad \left. + \Pr[1\leftarrow {\Dc'}\left(\ketbra{\Phi^-}\otimes \rho_{0,\lambda}\right)] - \Pr[1\leftarrow {\Dc'}\left(\ketbra{\Phi^-}\otimes \rho_{1,\lambda}\right)] \right|.
    \end{split}\end{align}
    Here, note that for each \(\ketbra{\Phi^+}\) and \(\ketbra{\Phi^-}\), there are two differences of the probabilities of \(\Dc'\) outputting \(1\) for each EFI pair element \(\rho_{0,\lambda}\) and \(\rho_{1,\lambda}\).
    Then, by using the triangle inequality, 
    \begin{align}\begin{split} \label{eq:triang}
        \Adv_{\Dc'}(\psi_{AB}^\lambda,\phi_{AB}^\lambda) 
        &\leq \frac{1}{4} \left(\left| -\Pr[1\leftarrow {\Dc'}\left(\ketbra{\Phi^+}\otimes \rho_{0,\lambda}\right)] + \Pr[1\leftarrow {\Dc'}\left(\ketbra{\Phi^+}\otimes \rho_{1,\lambda}\right)] \right| \right.\\
        &\qquad\qquad +\left. \left| \Pr[1\leftarrow {\Dc'}\left(\ketbra{\Phi^-}\otimes \rho_{0,\lambda}\right)] - \Pr[1\leftarrow {\Dc'}\left(\ketbra{\Phi^-}\otimes \rho_{1,\lambda}\right)] \right|\right).
    \end{split}\end{align}
    Clearly, any distinguisher \(\Dc\) that receives an unknown state \(\rho_{b,\lambda}\), for \(b\in\{0,1\}\), has the same advantage in distinguishing \(\rho_{0,\lambda}\) and \(\rho_{1,\lambda}\) than in distinguishing \(\ketbra{\Phi^+}\otimes \rho_{0,\lambda}\) and \(\ketbra{\Phi^+}\otimes \rho_{1,\lambda}\). Indeed, any distinguisher \(\Dc\) that receives \(\rho_{b,\lambda}\) may just prepare this uncorrelated register \(\ketbra{\Phi^+}\) and try to distinguish the new state \(\ketbra{\Phi^+}\otimes \rho_{b,\lambda}\) (analogously for \(\ketbra{\Phi^-}\otimes \rho_{0,\lambda}\) and \(\ketbra{\Phi^-}\otimes \rho_{1,\lambda}\)), so, for \(\Phi\in\{\Phi^+,\Phi^-\}\),
    \begin{equation}
        \Adv_{\Dc'}\left(\ketbra{\Phi}\otimes \rho_{0,\lambda}, \ketbra{\Phi}\otimes \rho_{1,\lambda}\right) = \Adv_{\Dc}\left(\rho_{0,\lambda}, \rho_{1,\lambda}\right).
    \end{equation}
    Therefore, one may write Equation~\eqref{eq:triang} in terms of the advantage of an algorithm \(\Dc\) of distinguishing \(\rho_{0,\lambda}\) and \(\rho_{1,\lambda}\) as
    \begin{align}\begin{split}
        \Adv_{\Dc'}(\psi_{AB}^\lambda,\phi_{AB}^\lambda) &\leq \frac{1}{4} \left(\Adv_{\Dc'}\left(\ketbra{\Phi^+}\otimes \rho_{0,\lambda}, \ketbra{\Phi^+}\otimes \rho_{1,\lambda}\right) \right.\\
        &\qquad\qquad +\left. \Adv_{\Dc'}\left(\ketbra{\Phi^-}\otimes \rho_{0,\lambda}, \ketbra{\Phi^-}\otimes \rho_{1,\lambda}\right)\right)\\
        &\leq \frac{1}{2}\Adv_{\Dc}\left(\rho_{0,\lambda},\rho_{1,\lambda}\right)\\
        &\leq \frac{1}{2}\,\varepsilon(\lambda).
    \end{split}\end{align}
    Therefore, since by assumption \(\varepsilon\) is a negligible function, there exists another negligible function \(\varepsilon' = \frac{1}{2}\,\varepsilon\), such that for all distinguishers \(\Dc'\), \(\Adv_{\Dc'}(\psi_{AB}^\lambda,\phi_{AB}^\lambda)\leq \varepsilon'(\lambda)\).
    This demonstrates that if \(\rho_{0,\lambda}\) and \(\rho_{1,\lambda}\) are computationally indistinguishable, then so must be \(\psi_{AB}^\lambda\) and \(\phi_{AB}^\lambda\).\\
    Finally, it still needs to be argued that \(\psi_{AB}^\lambda\) and \(\phi_{AB}^\lambda\) are indistinguishable under \(p(\lambda)\) many copies for any polynomial \(p\).
    However, note that a distinguisher may prepare \(\psi_{AB}^\lambda\) and \(\phi_{AB}^\lambda\) locally and efficiently (by creating one EPR pair and running the EFI generator algorithm that is efficient by definition), so, by the hybrid argument (Lemma~\ref{lemma:hybridarg}), if 
    \begin{equation}
        \Adv_{\Dc'}\left(\psi_{AB}^\lambda,\phi_{AB}^\lambda\right) \leq \varepsilon(\lambda),
    \end{equation}
    then,
    \begin{equation}
        \Adv_{\Dc'}\left({\psi_{AB}^\lambda}^{\otimes p(\lambda)},{\phi_{AB}^\lambda}^{\otimes p(\lambda)}\right) \leq p(\lambda)\, \varepsilon(\lambda),
    \end{equation}
    and \(p(\lambda)\, \varepsilon(\lambda)\) is a negligible function, since it is a polynomial function, \(p\), multiplied by a negligible function, \(\varepsilon\).

    This concludes the proof.
\end{proof}

\section{Efficient Distillation Using Quantum Keys}\label{sec:quantumk}
In Definition~\ref{def:pseudoentanglement4}, we consider that the distillable entanglement of the reference family \(\{\phi_{AB}^\lambda\}_\lambda\) may be inefficient (i.e., non-QPT algorithms are allowed to perform the LOCC protocol to distill), motivated by this family not being required to be efficient to prepare.
That is, we consider the property that regard the distillable entanglement of \(\{\phi_{AB}^\lambda\}_\lambda\) to require \(E_D^\varepsilon(\{\phi_{AB}^\lambda\}_\lambda)\geq d\).
However, in the original pseudo-entanglement definition of~\cite{pre:ABV23} (Definition~\ref{def:pseudoentanglement2}), in contrast with the definition of pseudo-entanglement of~\cite{pre:ABF+22}, it is considered that distillable entanglement of this family is efficiently accessible given a classical key \(k\in\{0,1\}^{\kappa(\lambda)}\) (for polynomial \(\kappa:\N\to\N\)), \(\hat{E}_D^\varepsilon(\{k_\lambda,\phi_{AB}^\lambda\}_\lambda)\geq d\).

In this section, we provide an adaptation where we allow the key \(k\) to be a quantum state, and show that, given this key, entanglement distillation may be made efficient.
Recall that an LOCC protocol given a quantum state \(\rho^k\) indexed by some key \(k\in \{0,1\}^{\kappa(\lambda)}\), \(\Gamma(k,\rho^k)\), is defined as the LOCC protocol \(\Gamma\) acting on the bipartite state \(\ketbra{k}_{A'}\otimes \rho_{AB}^k \otimes \ketbra{k}_{B'}\), for bipartition \(AA':BB'\).
We then allow the classical key \(k\) of Definition~\ref{def:pseudoentanglement2} to be a quantum state \(k\) that is given to party \(A\) or \(B\) (or possibly both).

First, the structures of both families of the pseudo-entanglement structure, \(\{\psi_{AB}^\lambda\}_\lambda\) and \(\{\phi_{AB}^\lambda\}_\lambda\), are the same.
However, we now consider that the EFI pair is generated by the unitary part \(\hat{\As}_{b,\lambda}\) of the EFI generator \(\As(1^\lambda,b)\).
This construction is inspired by the construction of canonical commitments from EFI pairs of~\cite{ITCS:BCQ23}.
Let \(\lambda\in\N\), denote the full coherent state output by running \(\hat{\As}_{b,\lambda}\ket{0}\) as
\begin{equation}
    \ket{\chi_{b,\lambda}}_{EK} = \hat{\As}_{b,\lambda}\ket{0}.
\end{equation}
Moreover, consider the subspace \(E\) to be the ``EFI pair space'', and the subspace \(K\) to be the ``key space''.
I.e., 
\begin{align}
    \rho_{b,\lambda} &= \Tr_{K}(\ketbra{\chi_{b,\lambda}}) = \As(1^\lambda, b),\\
    k_{b,\lambda} &= \Tr_{E}(\ketbra{\chi_{b,\lambda}}).
\end{align}

Construct \(\psi_{AB}^\lambda\) and \(\phi_{AB}^\lambda\) as before,
\begin{align}
    \psi_{AB}^\lambda &= \frac{1}{4}\left(\ketbra{\Phi^+}_{AB} + \ketbra{\Phi^-}_{AB}\right)\otimes \left({\rho_{0,\lambda}}_A + {\rho_{1,\lambda}}_A\right),\\
    \phi_{AB}^\lambda &= \frac{1}{2}\left(\ketbra{\Phi^+}_{AB}\otimes {\rho_{0,\lambda}}_A + \ketbra{\Phi^-}_{AB}\otimes {\rho_{1,\lambda}}_A \right).
\end{align}

For our purpose, the key \(k_\lambda\) may be considered to be \(k_\lambda = k_{0,\lambda}\) or \(k_\lambda = k_{1,\lambda}\).
Without loss of generality, we assume \(k_\lambda = k_{0,\lambda}\).
Also, note that the state \(k_\lambda\) is kept coherent, and we assume that this state (the key) is given to the party holding the EFI pair, in this case, party \(A\).

Clearly, the computational entanglement cost of \(\psi_{AB}^\lambda\) is still zero (the states are the same as in Equations~\eqref{eq:psi} and~\eqref{eq:phi}),
\begin{equation}
    \hat{E}_C^0\left(\left\{k_\lambda, \psi_{AB}^\lambda\right\}_\lambda\right)=0,
\end{equation} as \(k\) may just be ignored.

Also, computational indistinguishability of states \(\psi_{AB}^\lambda\) and \(\phi_{AB}^\lambda\) still holds, as the key \(k\) is not given to the distinguisher, thus, the states being distinguished are exactly the same as in the proof of computational indistinguishability of  Theorem~\ref{thm:efi2pe}.
Therefore, \({\psi_{AB}^\lambda}^{\otimes p(\lambda)}\) and \({\phi_{AB}^\lambda}^{\otimes p(\lambda)}\) are computationally indistinguishable.

More interesting is the distillable entanglement property, namely the fact that, given \(k_\lambda\), there exists an efficient LOCC protocol that is able to recover \(1\) EPR pair from \({\phi_{AB}^\lambda}\) with exponential small error \(O(2^{-\lambda})\), i.e., 
\begin{equation}
    \hat{E}_D^{2^{-\lambda}}\left(\{k_\lambda,\phi_{AB}^\lambda\}_\lambda\right)=1.
\end{equation}
To accomplish this, party \(A\) will use the key \(k_\lambda\) to distinguish the EFI pair elements \(\rho_{0,\lambda}\) and \(\rho_{1,\lambda}\) as follows:
\begin{enumerate}
    \item Given \(k_\lambda = k_{0,\lambda}\) and \(\phi_{AB}^\lambda\) as input, consider the second register of \(\phi_{AB}^\lambda\) to be \(\rho_{b,\lambda}\), for unknown \(b\in\{0,1\}\).
    \item Construct the state \(\chi_{b,\lambda}'\) from the given unknown EFI pair element (\(\rho_{b,\lambda}\)) and given key \(k_\lambda\). If \(b=0\), then \(\chi_{0,\lambda}'=\ketbra{\chi_{0,\lambda}}\) (the original coherent state is reconstructed); if \(b=1\), then \(\chi_{1,\lambda}'\) is a product state (across partition \(E:K\)).
    \item Compute \({\xi} = {\hat{\As}_{0,\lambda}}^\dagger {\chi_{b,\lambda}'} \hat{\As}_{0,\lambda}\). If \(b=0\), then \(\xi =  {\hat{\As}_{0,\lambda}}^\dagger  \ketbra{\chi_{0,\lambda}} \hat{\As}_{0,\lambda} =\ketbra{0}\).
    \item Measure each qubit \(i\) of the state \({\xi}\) individually in the computational basis and store the result \(m_i\).
    \item Return \(0\) if \(m_i = 0\) for all \(i\); otherwise, return \(1\).
\end{enumerate}
This algorithm runs in QPT and distinguishes \(\rho_{0,\lambda}\) and \(\rho_{1,\lambda}\) with overwhelming probability since \(\TD(\rho_{0,\lambda},\rho_{1,\lambda}) \geq 1-2^{-\lambda}\), thus the overlap (i.e., the fidelity) between \(\rho_{0,\lambda}\) and \(\rho_{1,\lambda}\) is exponentially small (\(\Fid(\rho_{0,\lambda},\rho_{1,\lambda})\leq \sqrt{2}\;2^{-\lambda/2}\)).
Finally, in analogy with the proof of the distillable entanglement in Theorem~\ref{thm:efi2pe}, given that the party holding the system \(A\) now knows that it has \(\rho_{0,\lambda}\) or \(\rho_{1,\lambda}\) in the second register, it also knows which is the EPR pair in the first register (\(\ket{\Phi^+}_{AB}\) or \(\ket{\Phi^-}_{AB}\)).
This EPR pair is shared with the party holding system \(B\), and the party holding \(A\) may communicate this information classically, such that in the end they both output this known EPR pair with overwhelming probability in \(\lambda\).

\section{Polynomial Amplification of Pseudo-Entanglement}\label{sec:amplification}
In this section, we show that given a family of  \(c\) vs.\ \(d\) (according to Definition~\ref{def:pseudoentanglement})  or \((\varepsilon,c,d)\) (according to Definition~\ref{def:pseudoentanglement2}) pseudo-entangled states, one can amplify the gap (\(d-c\)) by a polynomial factor \(q\), i.e., \((d-c)\, q(\lambda)\) (Lemma~\ref{lemma:amplify}).%
    \footnote{For the sake of generality, we choose to formulate this result for the settings where, for each security parameter \(\lambda\), there may exist multiple states indexed by some secret or public key \(k\) (settings of Definitions~\ref{def:pseudoentanglement} and~\ref{def:pseudoentanglement2}). This setting trivially encompasses the setting where there is only one state for each \(\lambda\) (e.g., Definitions~\ref{def:pseudoentanglement3} and~\ref{def:pseudoentanglement4}) by setting \(k=\lambda\).}
This result follows from the application of a hybrid argument to the repeated sampling of the low-entanglement and (the not necessarily efficient sampling of) high-entangled families, and direct tensor of the quantum states.
Indeed, from these new tensor states, the parties may act independently on each state, such that the entanglement of a pseudo-entangled family with entanglement \(c\) is increased to \(c\,  q(\lambda)\), while remaining computationally indistinguishable from another family with entanglement \(d\, q(\lambda)\), for any polynomial \(q\).
Indeed, the new families are still indistinguishable by a hybrid argument, where the probability of error (i.e., the loss in fidelity) increases polynomially since the samples are independent.
We remark the subtle technicality that \(\phi_{AB}^\lambda\) may not be efficiently preparable (even locally) by a uniform QPT algorithm.
However, we consider computational indistinguishability against non-uniform adversaries (Definition~\ref{def:computationalind}).
Therefore, the polynomially many copies of the states \(\phi_{AB}^\lambda\) may be given to the distinguisher as advice (quantum states with polynomial number of qubits) such that the computational indistinguishability under polynomially many copies still holds by a hybrid argument (Lemma~\ref{lemma:hybridarg}).

Applying this result directly to the result from Theorem~\ref{thm:efi2pe} allows us to show that polynomial-gap pseudo-entanglement is necessary for EFI pairs (Corollary~\ref{cor:efi2pe-poly}).

\begin{lemma}\label{lemma:amplify}
    Let \(\lambda\in\mathbb{N}\), \(\kappa:\mathbb{N}\to\mathbb{N}\) a polynomially bounded function. 
    Let \(\left\{k,\psi_{AB}^k\right\}_{k\in\{0,1\}^{\kappa(\lambda)}}\) and \(\left\{k,\phi_{AB}^k\right\}_{k\in\{0,1\}^{\kappa(\lambda)}}\) be two families of states, with \(\hat{E}_C^\varepsilon\left(\left\{k,\psi_{AB}^k\right\}_k\right)\leq c(\lambda)\) (also for any \(E^\varepsilon\) such that \(E^\varepsilon\leq E_C^\varepsilon\leq \hat{E}_C^\varepsilon\)), and \(E_D^\varepsilon\left(\left\{k,\phi_{AB}^k\right\}_k\right)\geq d(\lambda)\) (also for any \(E^\varepsilon\) such that \(E_D^\varepsilon\leq E^\varepsilon\)), and \(\mathbb{E}_k\left[{\psi_{AB}^k}^{\otimes p(\lambda)}\right]\) and \(\mathbb{E}_k\left[{\phi_{AB}^k}^{\otimes p(\lambda)}\right]\) are computationally indistinguishable.
    For any polynomial function \(q\), there exist families \(\left\{\overline{k},\overline{\psi}_{AB}^{\overline{k}}\right\}_{\overline{k}\in\{0,1\}^{q(\lambda)\kappa(\lambda)}}\) and \(\left\{\overline{k},\overline{\phi}_{AB}^{\overline{k}}\right\}_{\overline{k}\in\{0,1\}^{q(\lambda)\kappa(\lambda)}}\), such that \(\hat{E}_C^{\overline{\varepsilon}}\left(\left\{\overline{k},\overline{\psi}_{AB}^{\overline{k}}\right\}_{\overline{k}}\right)\leq c(\lambda) q(\lambda)\) (also for any \(E^{\overline{\varepsilon}}\)  such that \(E^{\overline{\varepsilon}}\leq E_C^{\overline{\varepsilon}}\leq \hat{E}_C^{\overline{\varepsilon}}\)),  and \({E}_D^{\overline{\varepsilon}}\left(\left\{\overline{k},\overline{\phi}_{AB}^{\overline{k}}\right\}_{\overline{k}}\right)\geq d(\lambda) q(\lambda)\) (also for any \(E^{\overline{\varepsilon}}\) such that \(E_D^{\overline{\varepsilon}} \leq E^{\overline{\varepsilon}}\)), with \({\overline{\varepsilon}} = \varepsilon(\lambda) q(\lambda)\), and \(\mathbb{E}_{\overline{k}}\left[{\overline{\psi}_{AB}^{\overline{k}}}^{\otimes p(\lambda)}\right]\) and \(\E_{\overline{k}}\left[{\overline{\phi}_{AB}^{\overline{k}}}^{\otimes p(\lambda)}\right]\) are computationally indistinguishable, for any polynomial \(p\).
\end{lemma}
\begin{proof}
    Let the two new families \(\left\{{\overline{k}}, \overline{\psi}_{AB}^{\overline{k}}\right\}_{\overline{k}}\) and \(\left\{{\overline{k}}, \overline{\phi}_{AB}^{\overline{k}}\right\}_{\overline{k}}\) be constructed from the given families \(\left\{k,\psi_{AB}^k\right\}_k\) and \(\left\{k,\phi_{AB}^k\right\}_k\) as
    \begin{align}
    \overline{\psi}_{AB}^{\overline{k}} &= \bigotimes_{i=1}^{q(\lambda)} \psi_{AB}^{k_i},\\
    \overline{\phi}_{AB}^{\overline{k}} &= \bigotimes_{i=1}^{q(\lambda)} \phi_{AB}^{k_i},
    \end{align}
    with \({\overline{k}}=k_1 || k_2 ||\dots || k_{q(\lambda)}\).
    Both constructions are still polynomial size, as they are composed of polynomially many repetitions of elements of polynomial length.
    Moreover, since all \(q(\lambda)\) elements of the families are independent, the parties may execute the LOCC protocols iteratively (sample-by-sample) for each \(k_i\) to construct \(\overline{\psi}_{AB}^{\overline{k}}\) with computational entanglement cost 
    \begin{equation}
    \hat{E}_C^{\overline{\varepsilon}}\left(\left\{{\overline{k}},\overline{\psi}_{AB}^{\overline{k}}\right\}_{\overline{k}}\right)\leq c(\lambda) q(\lambda);
    \end{equation}
    while the distillable entanglement of \(\overline{\phi}_{AB}^{\overline{k}}\) is
    \begin{equation}
        {E}_D^{\overline{\varepsilon}}\left(\left\{{\overline{k}},\overline{\phi}_{AB}^{\overline{k}}\right\}_{\overline{k}}\right)\geq d(\lambda) q(\lambda).
    \end{equation}
    Likewise, the new error \(\overline{\varepsilon}\) also scales polynomially.
    Indeed, the fidelity of two tensored states \(\Fid(\rho_1 \otimes \rho_2, \sigma_1 \otimes \sigma_2) = \Fid(\rho_1,\sigma_1) \Fid(\rho_2,\sigma_2)\), then for \(q\) tensored states, \(\Fid(\rho^{\otimes q}, \sigma^{\otimes q}) = \Fid(\rho,\sigma)^q\).
    So, for the entanglement cost (Definition~\ref{def:ec-it})
    \begin{align} \begin{split}
        \Fid\left({\Gamma\left(\Phi^{\otimes c}\right)}^{\otimes q(\lambda)}, {\overline{\psi}_{AB}^{\overline{k}}}\right) = \Fid\left(\Gamma\left(\Phi^{\otimes c}\right),  {\psi_{AB}^k}\right)^{q(\lambda)} &\geq (1- \varepsilon(\lambda))^{q(\lambda)}\\
        &\geq 1-\varepsilon(\lambda)q(\lambda).
    \end{split} \end{align}    
    Similarly, for the distillable entanglement (Definition~\ref{def:ed-it}),
    \begin{align} \begin{split}
        \Fid\left(\Gamma\left({\overline{\phi}_{AB}^{\overline{k}}}\right),{\left(\Phi^{\otimes d}\right)}^{\otimes q(\lambda)}\right)  = \Fid\left(\Gamma\left({{\phi}_{AB}^{{k}}}\right),\Phi^{\otimes d}\right)^{q(\lambda)} &\geq (1- \varepsilon(\lambda))^{q(\lambda)}\\
        &\geq 1-\varepsilon(\lambda)q(\lambda).
    \end{split} \end{align}
    Where the last inequalities follow from Bernoulli's inequality, since \(q(\lambda)\) is an integer larger than 1 (number of copies), and \(\varepsilon(\lambda)\in [0,1]\).

    It remains to show the computational indistinguishability of the states for polynomially many copies (\(p(\lambda)\)), \(\mathbb{E}_{\overline{k}}\left[{\overline{\psi}_{AB}^{\overline{k}}}^{\otimes p(\lambda)}\right]\) and \(\mathbb{E}_{\overline{k}}\left[{\overline{\phi}_{AB}^{\overline{k}}}^{\otimes p(\lambda)}\right]\).
    Here, note that 
    \begin{align}
        {\overline{\psi}_{AB}^{\overline{k}}}^{\otimes p(\lambda)} = \left({\bigotimes_{i=1}^{q(\lambda)} \psi_{AB}^{k_i}}\right)^{\otimes p(\lambda)} \equiv \bigotimes_{i=1}^{q(\lambda)} \left({\psi_{AB}^{k_i}}^{\otimes p(\lambda)}\right),\label{eq:psibar}\\
        {\overline{\phi}_{AB}^{\overline{k}}}^{\otimes p(\lambda)} = \left({\bigotimes_{i=1}^{q(\lambda)} \phi_{AB}^{k_i}}\right)^{\otimes p(\lambda)} \equiv \bigotimes_{i=1}^{q(\lambda)} \left({\phi_{AB}^{k_i}}^{\otimes p(\lambda)}\right),\label{eq:phibar}
    \end{align}
    where the relation \(\equiv\) represents an equivalent state where the registers are permuted such that the copies of \(\psi_{AB}^{k_i}\) and \(\phi_{AB}^{k_i}\) for the same \(k_i\) are grouped together.
    But now, for each \(i\), computational indistinguishability is guaranteed by assumption (indistinguishability of the original families \(\psi_{AB}^{k}\) and \(\phi_{AB}^{k}\)), and all it is being done is taking polynomially many independent elements of the two computationally indistinguishable original families of states.
    I.e.,
    \begin{equation}\begin{alignedat}{2}
        \mathbb{E}_{\overline{k}}\left[{\overline{\psi}_{AB}^{\overline{k}}}^{\otimes p(\lambda)}\right] &\equiv \mathbb{E}_{\overline{k}}\left[\bigotimes_{i=1}^{q(\lambda)} \left({\psi_{AB}^{k_i}}^{\otimes p(\lambda)}\right)\right] &\qquad\qquad& \text{(by Equation~\eqref{eq:psibar})}\\
        &= \bigotimes_{i=1}^{q(\lambda)} \mathbb{E}_{k_i}\left[{\psi_{AB}^{k_i}}^{\otimes p(\lambda)}\right] && \text{(by independence of each \(k_i \in \bar{k}\))}\\
        &\approx \bigotimes_{i=1}^{q(\lambda)} \mathbb{E}_{k_i}\left[{\phi_{AB}^{k_i}}^{\otimes p(\lambda)}\right] && \text{(by Lemma~\ref{lemma:hybridarg})}\\
        &=\mathbb{E}_{\overline{k}}\left[\bigotimes_{i=1}^{q(\lambda)} \left({\phi_{AB}^{k_i}}^{\otimes p(\lambda)}\right)\right] && \text{(by independence of each \(k_i \in \bar{k}\))}\\
        &\equiv \mathbb{E}_{\overline{k}}\left[{\overline{\phi}_{AB}^{\overline{k}}}^{\otimes p(\lambda)}\right] && \text{(by Equation~\eqref{eq:phibar})}
    \end{alignedat}\end{equation}
    where the relation $\approx$ means computational indistinguishability.
    Indeed, the computational indistinguishability holds by the hybrid argument since the original families are computationally indistinguishable against all non-uniform QPT adversaries, and therefore the application of the hybrid argument is valid (Lemma~\ref{lemma:hybridarg}), where the states  \(\E_{\overline{k}}\left[{\overline{\psi}_{AB}^{\overline{k}}}^{\otimes p(\lambda)}\right]\), \(\E_{\overline{k}}\left[{\overline{\phi}_{AB}^{\overline{k}}}^{\otimes p(\lambda)}\right]\) are given as the advice in the usual reduction (from a distinguisher of many samples to a distinguisher of a single sample).
    Therefore, for all \(\Dc,\Dc'\), there always exist two negligible functions \(\varepsilon,\varepsilon'\) such that
    \begin{align}\begin{split}
        &\Adv_{\Dc}\left(\mathbb{E}_{{k}}\left[{{\psi}_{AB}^{{k}}}^{\otimes p(\lambda)}\right],\mathbb{E}_{{k}}\left[{{\phi}_{AB}^{{k}}}^{\otimes p(\lambda)}\right]\right)\leq \varepsilon(\lambda) \Rightarrow\\
        &\qquad \Adv_{\Dc'}\left(\mathbb{E}_{\overline{k}}\left[{\overline{\psi}_{AB}^{\overline{k}}}^{\otimes p(\lambda)}\right],\mathbb{E}_{\overline{k}}\left[{\overline{\phi}_{AB}^{\overline{k}}}^{\otimes p(\lambda)}\right]\right) \leq q(\lambda)\,\varepsilon(\lambda) = \varepsilon'(\lambda).
    \end{split}\end{align}
\end{proof}

\begin{corollary}\label{cor:efi2pe-poly}
    Let \(\lambda\in\mathbb{N}\), \(n:\mathbb{N}\to\mathbb{N}\) a polynomially bounded function, and \(\As\) a QPT algorithm generating an EFI pair as defined in Definition~\ref{def:efi}.
    There exist two families of \((q(\lambda)\, n(\lambda)+q(\lambda))\)-qubit bipartite quantum states, \(\{\psi_{AB}^\lambda\}_\lambda\) with \(\hat{E}_C^\varepsilon\left(\{\psi_{AB}^\lambda\}_\lambda\right) = 0\) (also for any \(E^\varepsilon\) such that \(E^\varepsilon\leq E_C^\varepsilon\leq \hat{E}_C^\varepsilon\)), and \(\{\phi_{AB}^\lambda\}_\lambda\) with \({E}_D^\varepsilon\left(\{\phi_{AB}^\lambda\}_\lambda\right) \geq q(\lambda)\) (also for any \(E^\varepsilon\) such that \(E_D^\varepsilon\leq E^\varepsilon\)), for any polynomial function \(q\), with \(\varepsilon \in O(2^{-\lambda})\), such that \({\psi_{AB}^\lambda}^{\otimes p(\lambda)}\) and \({\phi_{AB}^\lambda}^{\otimes p(\lambda)}\) are computationally indistinguishable, for any polynomial \(p\).
\end{corollary}
\begin{proof}
    The proof follows directly from Theorem~\ref{thm:efi2pe} and Lemma~\ref{lemma:amplify}.
\end{proof}

\bibliographystyle{alpha}
\bibliography{bib}

\newcommand{\etalchar}[1]{$^{#1}$}
\begin{thebibliography}{MMWY24}

\bibitem[ABF{\etalchar{+}}24]{pre:ABF+22}
Scott Aaronson, Adam Bouland, Bill Fefferman, Soumik Ghosh, Umesh Vazirani, Chenyi Zhang, and Zixin Zhou.
\newblock {Quantum Pseudoentanglement}.
\newblock In Venkatesan Guruswami, editor, {\em 15th Innovations in Theoretical Computer Science Conference (ITCS 2024)}, volume 287 of {\em Leibniz International Proceedings in Informatics (LIPIcs)}, pages 2:1--2:21, Dagstuhl, Germany, 2024. Schloss Dagstuhl -- Leibniz-Zentrum f{\"u}r Informatik.

\bibitem[ABV23]{pre:ABV23}
Rotem {Arnon-Friedman}, Zvika Brakerski, and Thomas Vidick.
\newblock Computational entanglement theory, 2023.
\newblock \url{https://arxiv.org/abs/2310.02783}.

\bibitem[AQY22]{C:AQY22}
Prabhanjan Ananth, Luowen Qian, and Henry Yuen.
\newblock Cryptography from pseudorandom quantum states.
\newblock In Yevgeniy Dodis and Thomas Shrimpton, editors, {\em Advances in Cryptology -- CRYPTO 2022}, pages 208--236, Cham, 2022. Springer Nature Switzerland.

\bibitem[BB14]{TCS:BB84}
Charles~H. Bennett and Gilles Brassard.
\newblock Quantum cryptography: Public key distribution and coin tossing.
\newblock {\em Theoretical Computer Science}, 560:7–11, December 2014.

\bibitem[BCKM21]{C:BCKM21}
James Bartusek, Andrea Coladangelo, Dakshita Khurana, and Fermi Ma.
\newblock One-way functions imply secure computation in a quantum world.
\newblock In Tal Malkin and Chris Peikert, editors, {\em Advances in Cryptology -- CRYPTO 2021}, pages 467--496, Cham, 2021. Springer International Publishing.

\bibitem[BCQ23]{ITCS:BCQ23}
Zvika Brakerski, Ran Canetti, and Luowen Qian.
\newblock {On the Computational Hardness Needed for Quantum Cryptography}.
\newblock In Yael Tauman~Kalai, editor, {\em 14th Innovations in Theoretical Computer Science Conference (ITCS 2023)}, volume 251 of {\em Leibniz International Proceedings in Informatics (LIPIcs)}, pages 24:1--24:21, Dagstuhl, Germany, 2023. Schloss Dagstuhl -- Leibniz-Zentrum f{\"u}r Informatik.

\bibitem[BD10]{JMP:BD10}
Francesco Buscemi and Nilanjana Datta.
\newblock {Distilling entanglement from arbitrary resources}.
\newblock {\em Journal of Mathematical Physics}, 51(10):102201, 10 2010.

\bibitem[BD11]{PRL:BD11}
Francesco Buscemi and Nilanjana Datta.
\newblock Entanglement cost in practical scenarios.
\newblock {\em Phys. Rev. Lett.}, 106:130503, Mar 2011.

\bibitem[BFG{\etalchar{+}}24]{pre:BFG+23}
Adam Bouland, Bill Fefferman, Soumik Ghosh, Tony Metger, Umesh Vazirani, Chenyi Zhang, and Zixin Zhou.
\newblock {Public-Key Pseudoentanglement and the Hardness of Learning Ground State Entanglement Structure}.
\newblock In Rahul Santhanam, editor, {\em 39th Computational Complexity Conference (CCC 2024)}, volume 300 of {\em Leibniz International Proceedings in Informatics (LIPIcs)}, pages 21:1--21:23, Dagstuhl, Germany, 2024. Schloss Dagstuhl -- Leibniz-Zentrum f{\"u}r Informatik.

\bibitem[BMB{\etalchar{+}}24]{pre:BMB+24}
Nikhil Bansal, Wai-Keong Mok, Kishor Bharti, Dax~Enshan Koh, and Tobias Haug.
\newblock Pseudorandom density matrices, 2024.
\newblock \url{https://arxiv.org/abs/2407.11607}.

\bibitem[DHR02]{JMP:DHR02}
Matthew~J. Donald, Michał Horodecki, and Oliver Rudolph.
\newblock {The uniqueness theorem for entanglement measures}.
\newblock {\em Journal of Mathematical Physics}, 43(9):4252--4272, 09 2002.

\bibitem[GB23]{pre:GB23}
Tudor {Giurgica-Tiron} and Adam Bouland.
\newblock Pseudorandomness from subset states, 2023.
\newblock \url{https://arxiv.org/abs/2312.09206}.

\bibitem[GH20]{pre:GH20}
Alexandru Gheorghiu and Matty~J. Hoban.
\newblock Estimating the entropy of shallow circuit outputs is hard, 2020.
\newblock \url{https://arxiv.org/abs/2002.12814}.

\bibitem[GIKL24]{pre:GIKL24}
Sabee Grewal, Vishnu Iyer, William Kretschmer, and Daniel Liang.
\newblock Pseudoentanglement ain't cheap, 2024.
\newblock \url{https://arxiv.org/abs/2404.00126}.

\bibitem[GLG{\etalchar{+}}24]{PRL:GLG+24}
Andi Gu, Lorenzo Leone, Soumik Ghosh, Jens Eisert, Susanne~F. Yelin, and Yihui Quek.
\newblock Pseudomagic quantum states.
\newblock {\em Phys. Rev. Lett.}, 132:210602, May 2024.

\bibitem[GLSV21]{EC:GLSV21}
Alex~B. Grilo, Huijia Lin, Fang Song, and Vinod Vaikuntanathan.
\newblock Oblivious transfer is in miniqcrypt.
\newblock In Anne Canteaut and Fran{\c{c}}ois-Xavier Standaert, editors, {\em Advances in Cryptology -- EUROCRYPT 2021}, pages 531--561, Cham, 2021. Springer International Publishing.

\bibitem[Gol90]{IPL:G90}
Oded Goldreich.
\newblock A note on computational indistinguishability.
\newblock {\em Inf. Process. Lett.}, 34(6):277–281, may 1990.

\bibitem[GS99]{JCSS:GS99}
Oded Goldreich and Madhu Sudan.
\newblock Computational indistinguishability: A sample hierarchy.
\newblock {\em Journal of Computer and System Sciences}, 59(2):253--269, 1999.

\bibitem[HBK24]{pre:HBK24}
Tobias Haug, Kishor Bharti, and Dax~Enshan Koh.
\newblock Pseudorandom unitaries are neither real nor sparse nor noise-robust, 2024.
\newblock \url{https://arxiv.org/abs/2306.11677}.

\bibitem[ILL89]{STOC:ILL89}
R.~Impagliazzo, L.~A. Levin, and M.~Luby.
\newblock Pseudo-random generation from one-way functions.
\newblock In {\em Proceedings of the Twenty-First Annual ACM Symposium on Theory of Computing}, STOC '89, page 12–24, New York, NY, USA, 1989. Association for Computing Machinery.

\bibitem[Imp95]{CT:I95}
R.~Impagliazzo.
\newblock A personal view of average-case complexity.
\newblock In {\em Proceedings of Structure in Complexity Theory. Tenth Annual IEEE Conference}, pages 134--147, 1995.

\bibitem[JLS18]{C:JLS18}
Zhengfeng Ji, Yi-Kai Liu, and Fang Song.
\newblock Pseudorandom quantum states.
\newblock In Hovav Shacham and Alexandra Boldyreva, editors, {\em Advances in Cryptology -- CRYPTO 2018}, pages 126--152, Cham, 2018. Springer International Publishing.

\bibitem[KQST23]{STOC:KQST23}
William Kretschmer, Luowen Qian, Makrand Sinha, and Avishay Tal.
\newblock Quantum cryptography in algorithmica.
\newblock In {\em Proceedings of the 55th Annual ACM Symposium on Theory of Computing}, STOC 2023, page 1589–1602, New York, NY, USA, 2023. Association for Computing Machinery.

\bibitem[Kre21]{ITCS:K21}
William Kretschmer.
\newblock {Quantum Pseudorandomness and Classical Complexity}.
\newblock In Min-Hsiu Hsieh, editor, {\em 16th Conference on the Theory of Quantum Computation, Communication and Cryptography (TQC 2021)}, volume 197 of {\em Leibniz International Proceedings in Informatics (LIPIcs)}, pages 2:1--2:20, Dagstuhl, Germany, 2021. Schloss Dagstuhl -- Leibniz-Zentrum f{\"u}r Informatik.

\bibitem[KT24]{pre:KT23}
Dakshita Khurana and Kabir Tomer.
\newblock Commitments from quantum one-wayness.
\newblock In {\em Proceedings of the 56th Annual ACM Symposium on Theory of Computing}, STOC 2024, page 968–978, New York, NY, USA, 2024. Association for Computing Machinery.

\bibitem[LMW24]{pre:LMW23}
Alex Lombardi, Fermi Ma, and John Wright.
\newblock A one-query lower bound for unitary synthesis and breaking quantum cryptography.
\newblock In {\em Proceedings of the 56th Annual ACM Symposium on Theory of Computing}, STOC 2024, page 979–990, New York, NY, USA, 2024. Association for Computing Machinery.

\bibitem[MMWY24]{pre:MMWY24}
Giulio Malavolta, Tomoyuki Morimae, Michael Walter, and Takashi Yamakawa.
\newblock Exponential quantum one-wayness and efi pairs, 2024.
\newblock \url{https://arxiv.org/abs/2404.13699}.

\bibitem[MY22]{C:MY22}
Tomoyuki Morimae and Takashi Yamakawa.
\newblock Quantum commitments and signatures without one-way functions.
\newblock In Yevgeniy Dodis and Thomas Shrimpton, editors, {\em Advances in Cryptology -- CRYPTO 2022}, pages 269--295, Cham, 2022. Springer Nature Switzerland.

\bibitem[MY23]{pre:MY23}
Tomoyuki Morimae and Takashi Yamakawa.
\newblock One-wayness in quantum cryptography, 2023.
\newblock \url{https://arxiv.org/abs/2210.03394}.

\bibitem[NC10]{NC10}
Michael~A. Nielsen and Isaac~L. Chuang.
\newblock {\em Quantum Computation and Quantum Information}.
\newblock Cambridge University Press, 2010.

\bibitem[Wil13]{W13}
Mark~M. Wilde.
\newblock {\em Quantum Information Theory}.
\newblock Cambridge University Press, 2013.

\bibitem[Yan22]{AC:Y22}
Jun Yan.
\newblock General properties of quantum bit commitments (extended abstract).
\newblock In Shweta Agrawal and Dongdai Lin, editors, {\em Advances in Cryptology -- ASIACRYPT 2022}, pages 628--657, Cham, 2022. Springer Nature Switzerland.

\end{thebibliography}

\end{document}